\documentclass[article,11pt]{article}

\usepackage{macros}
\usepackage{fullpage}

\title{Non-Linear Strong Data-Processing for \\ Quantum Hockey-Stick Divergences}

\author[1]{Theshani Nuradha}
\author[2]{Ian George}
\author[3]{Christoph Hirche}
\affil[1]{Department of Mathematics and
Illinois Quantum Information Science and Technology (IQUIST) Center,
University of Illinois Urbana-Champaign, Urbana, IL 61801, USA, 
theshani.gallage@gmail.com}
\affil[2]{Centre for Quantum Technologies,
National University of Singapore, Singapore 117543, Singapore,
qit.george@gmail.com}
\affil[3]{Institute for Information Processing (tnt/L3S),
Leibniz Universit\"at Hannover, Germany, christoph.hirche@gmail.com}

\date{ }
\begin{document}
\maketitle

\begin{abstract}
 Data-processing is a desired property of classical and quantum divergences and information measures. In information theory, the contraction coefficient measures how much the distinguishability of quantum states decreases when they are transmitted through a quantum channel, establishing linear strong data-processing inequalities (SDPI). However, these linear SDPI are not always tight and can be improved in most of the cases. In this work, we establish non-linear SDPI for quantum hockey-stick divergence for noisy channels that satisfy a certain noise criterion. We also note that our results improve upon existing linear SDPI for quantum hockey-stick divergences and also non-linear SDPI for classical hockey-stick divergence. We define $F_\gamma$ curves generalizing Dobrushin curves for the quantum setting while characterizing SDPI for the sequential composition of heterogeneous channels. In addition, we derive reverse-Pinsker type inequalities for $f$-divergences with additional constraints on hockey-stick divergences. 
 We show that these non-linear SDPI can establish tighter finite mixing times that cannot be achieved through linear SDPI. Furthermore, we find applications of these in establishing stronger privacy guarantees for the composition of sequential private quantum channels when privacy is quantified by quantum local differential privacy.
\end{abstract}

\tableofcontents

\section{Introduction}

The data-processing property is fundamental in both classical and quantum information processing. This property says that two states can only be less distinguishable after being processed by a (possibly noisy) channel. The data-processing property is formalized by the data-processing inequality of a divergence: a divergence $\sD$ satisfies the data-processing inequality (DPI) if for all channels $\cN$ and states $\rho$ and $\sigma$,
\begin{equation}\label{eq:DPI}
    \sD\!\left( \cN(\rho) \Vert \cN(\sigma) \right) \leq \sD(\rho \Vert \sigma).
\end{equation}
For the classical setting, the states $\rho$ and $\sigma$ can be replaced by two probability distributions $p$ and $q$ and $\cN$ by a Markov kernel (a conditional probability distribution).

The data processing inequality limits the useful information left in a system after undergoing a channel. As such, strengthened claims about the data processing inequality result in stronger claims about information processing tasks and are thus useful. In principle, there exists a fundamental function $f_{\sD}(\cN,\rho,\sigma) \coloneq \frac{\sD(\cN(\rho) \Vert \cN(\sigma))}{\sD(\rho \Vert \sigma)}$ so that for all $\cN,\rho,\sigma$,
\begin{align}\label{eq:opt-contr-coeff}
    \sD(\cN(\rho) \Vert \cN(\sigma)) = f_{\sD}(\cN,\rho,\sigma) \sD(\rho \Vert \sigma) \ .
\end{align}
This function could be useful as it would establish the strongest claims about data processing. Nonetheless, the problem with this function is two-fold. First, constructing this function for a fixed channel $\cN$ is likely equivalent to computing its value at every pair of inputs $(\rho,\sigma)$,
which is computationally too expensive. Second, in information-theoretic studies, we are primarily interested in claims that are more generic than appealing to specific inputs, and so this function is more specified than is actually useful. For both reasons, the aim is to find relaxations of this function that still obtain tighter results than \cref{eq:DPI}.

To that end, contraction coefficients of divergences have been studied extensively~\cite{Dobrushin1956,Hiai15, Hirche2022contraction,gao2022complete, asoodeh2023contractionegammadivergenceapplicationsprivacy,nuradha2024contraction,hirche2025partial,george2025unifiedapproachquantumcontraction,George-2024ergodic}. Formally, given a divergence $\sD$ satisfying DPI, the contraction coefficient $\eta_{\sD}(\cN) \in [0,1]$ is defined as
\begin{equation}
    \eta_{\sD}(\cN) \coloneq \sup_{\rho, \sigma \textnormal{ s.t.} \sD(\rho \Vert \sigma) \neq 0} \frac{\sD\!\left( \cN(\rho) \Vert \cN(\sigma) \right)}{\sD(\rho \Vert \sigma) } \ . \label{eq:SDPI_linear_opt}
\end{equation}
Using this definition, one sees that for all states $\rho$ and $\sigma$,
\begin{equation}
    \sD\!\left( \cN(\rho) \Vert \cN(\sigma) \right) \leq \eta_{\sD}(\cN) \ \sD(\rho \Vert \sigma). \label{eq:SDPI_linear}
\end{equation}
It follows immediately that $\eta_{\sD}(\cN)$ is an upper bound on $f_{\sD}(\cN,\rho,\sigma)$ for all pairs of inputs $\rho$ and $\sigma$ while \cref{eq:SDPI_linear} can only ever being a tighter inequality than \cref{eq:DPI}. Because of that, \cref{eq:SDPI_linear} is known as a strong data-processing inequality (SDPI) when $\eta_{\sD}(\cN) < 1$. Because $\eta_{\sD}(\cN)$ is a constant independent of $\rho$ and $\sigma$, \cref{eq:SDPI_linear} contracts the value of the divergence linearly, and thus we call \cref{eq:SDPI_linear} a \textit{linear} SDPI for divergence $\sD$.

While contraction coefficients $\eta_{\sD}(\cN)$ resolve the need for a more generically applicable quantity than $f_{\sD}(\cN,\rho,\sigma)$, they still suffer from certain issues. The first issue is that determining $\eta_{\sD}(\cN)$ is also generally computationally difficult. For example, for $\sD$ being the trace distance and $\cN$ being a quantum channel, it has been established to be an NP hard problem \cite{delsol2025computationalaspectstracenorm}. To circumvent this challenge, efficiently computable alternatives have been proposed to bound the contraction coefficients to obtain useful insights into information contraction \cite{Doeblin1937,makur2024doeblin,hirche2024quantum, george2025quantumdoeblincoefficientsinterpretations}. The second issue is that while linear SDPI provide useful improvements over the data-processing inequality, by definition it may be the case that a linear SDPI is only achieved for two specific inputs, and thus \eqref{eq:SDPI_linear} is not tight for almost all inputs. 

To improve upon the fact a linear SDPI may be loose for almost all inputs, one can establish \textit{non-linear} strong data processing inequalities~\cite{Polyanskiy-2015a}. This is done via a function that we refer to as the divergence-curve for a channel $\cN$:
\begin{equation} \label{eq:divergence_curve}
    F_{\sD}(\cN,t) \coloneq \sup\{\sD(\cN(\rho)\Vert \cN(\sigma)): \sD(\rho \Vert \sigma) \leq t\} \ . 
\end{equation}
From the definition of the divergence curve, a non-linear SDPI immediately follows:
\begin{equation}\label{eq:NL-SDPI}
     \sD\!\left( \cN(\rho) \Vert \cN(\sigma) \right) \leq F_{\sD}\left(\cN, \sD\!\left(\rho \Vert \sigma\right)\right) \ .
\end{equation}
\cref{eq:NL-SDPI} is referred to as a \textit{non-linear} SDPI specifically because the divergence contracts following the divergence-\textit{curve}. This inequality is appealing for two reasons. First, the divergence curve depends on $\cN$ and the value of $\sD(\rho \Vert \sigma)$, so any generically tighter inequality than \cref{eq:NL-SDPI} would rely upon knowledge of further structure on the inputs. Thus, \cref{eq:NL-SDPI} is the tightest inequality that remains agnostic about the inputs. Second, once $\cN$ is fixed, for every value the distinguishability $\sD\!\left(\rho \Vert \sigma \right)$ can take, the inequality is tight for at least one pair of inputs. Unless the information loss is independent of the initial distinguishability of the states, this improves upon the linear contraction coefficient. In fact, because $\eta_{\sD}(\cN)$ may be tight only on the worst-case value of $\sD(\rho \Vert \sigma)$, even upper bounds on $F_{\sD}(\cN,t)$ at a given $t$ can result on tighter bounds than applying a linear SDPI.  As tighter bounds obtain stronger claims about information processing, characterizing the divergence-curve is of both fundamental and practical relevance. 

In the classical setting, non-linear SDPI have been established for various divergences and information measures including total-variation distance (Dobrushin curves)~\cite{Polyanskiy-2015a}, mutual information ($F_I$ curves)~\cite{du2017strong}, Doeblin curves~\cite{lu2024doeblin} and hockey-stick divergences ($F_\gamma$ curves)~\cite{zamanlooy2024mathrm}. A closely related topic is the joint range approach for general $f$-divergences~\cite{harremoes2011pairs}.
In the quantum setting, only a generalization of the Dobrushin curve to obtain non-linear SDPI for the trace distance has been proposed in~\cite{huber2019jointly}.

Recently, the quantum hockey-stick divergence $E_{\gamma}(\rho \Vert \sigma)$~\cite{sharma2012strongconversesquantumchannel} has found new uses in quantum information processing beyond its initial application for strong converse bounds. In~\cite{hirche2024quantumDivergences}, a new family of quantum $f$-divergences defined via integrating the quantum hockey stick divergence over the parameter $\gamma \geq 1$ was introduced. By appealing to the contraction coefficient of the hockey-stick divergence, the authors were able to bound the contraction of these new $f$-divergences, which includes the contraction coefficient of the quantum relative entropy. Furthermore, the contraction coefficient of the quantum hockey-stick divergence has found applications in ensuring privacy for quantum systems~\cite{hirche2022quantum,nuradha2023quantum,angrisani2023differentialprivacyamplificationquantum, dasgupta2025quantum,gallage2025theory}. Given these recent results, it is clear that establishing non-linear SDPI for hockey-stick divergences will have immediate applications in privacy and mixing of quantum channels that satisfy certain noise criteria. However, non-linear SDPI for quantum divergences has not been studied beyond the Dobrushin curve. As such, the goal of this work is to establish non-linear SDPI for hockey-stick divergences and showcase improvements of non-linear SDPI over linear SDPI. 
We also suspect that these characterizations will provide insights into the non-linear SDPI for other families of divergences that stem from hockey-stick divergences~\cite{hirche2024quantumDivergences}.

\subsection{Contributions}
The main contributions of this work are as follows:

\begin{itemize}
    \item We obtain a non-linear strong data-processing inequality for the hockey-stick divergence in~\cref{thm:non_linear_HS_div} for channels that satisfy the following relation: for some $\gamma \geq 1$ and $\delta \in [0,1]$
    \begin{equation} \label{eq:intro_class_of_channels}
         \sup_{\rho, \sigma \in \cD(\cH)} E_{\gamma}\!\left( \cN(\rho) \Vert \cN(\sigma) \right) \leq \delta,
    \end{equation}
where $E_\gamma(\rho \Vert \sigma) \coloneq \Tr\!\left[ (\rho-\gamma \sigma)_+ \right]$.
To obtain the said non-linear characterization, we first derive a linear contraction coefficient for hockey-stick divergence in~\cref{prop:contraction_coeff_upper_bound} that holds for all $\gamma \geq 1$ and $\delta \in [0,1]$. As the next step, we obtain a general result on the strong data-processing of hockey-stick divergence in~\cref{lem:contraction_separate_gamma} with respect to a contraction coefficient that depends on the input states given. By utilizing both of these, we arrive at a \textit{non-linear} strong data-processing inequality for hockey-stick divergence. Figure~\ref{fig:compare} graphically shows the potential gains we can achieve from both linear and non-linear SDPI established in our work for a specific instance. This characterization generalizes the classical results obtained in~\cite{zamanlooy2023strong} to the quantum setting, as well as to the scenario where $\delta >0$. In fact, we show that this characterization is tight by finding a channel that achieve the said inequality.

\textit{Note that our results are valid for the classical setting by replacing states with probability distributions $p,q$ such that $\rho =\sum_{x} p(x) |x\rangle\! \langle x|$ and $\sigma =\sum_{x} q(x) |x\rangle\! \langle x|$.}

\begin{figure}[ht]
    \centering
\includegraphics[width=0.7\linewidth]{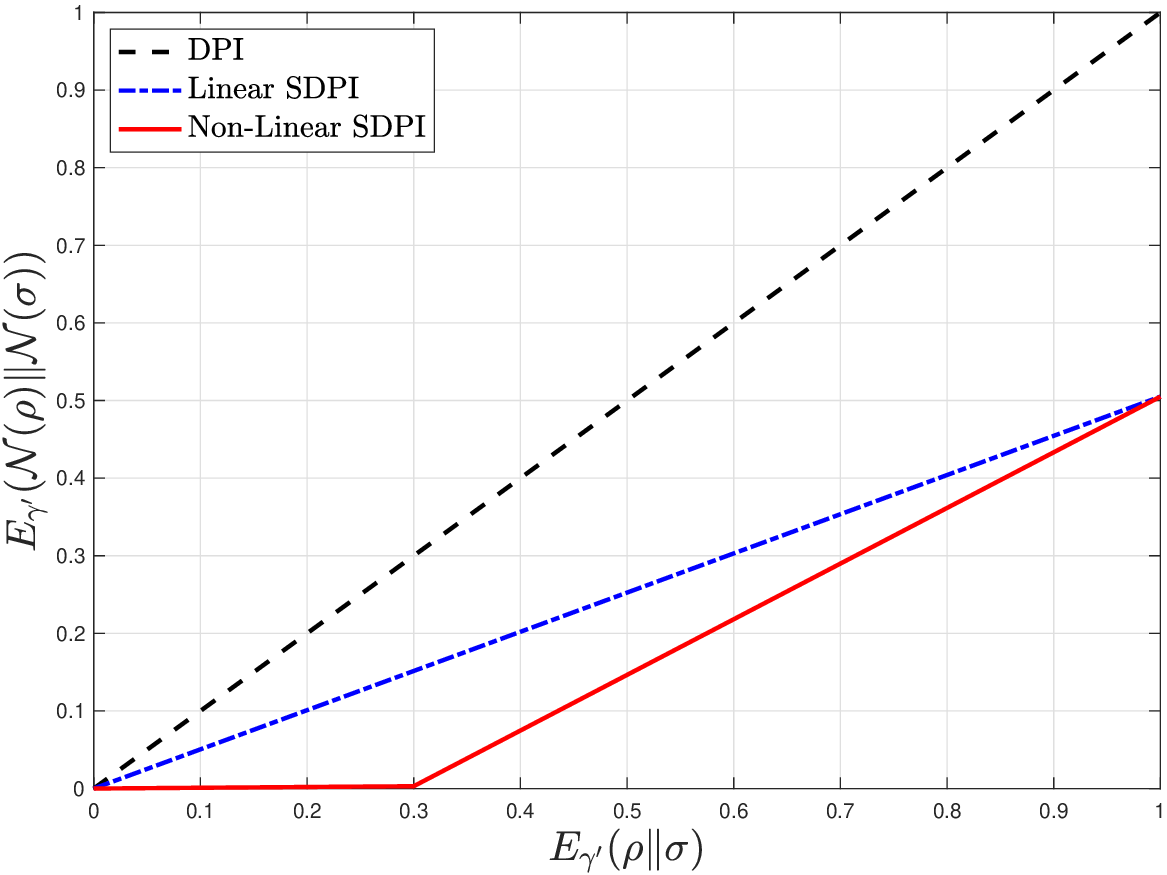}
    \caption{Comparison of data processing inequalities: DPI refers to the standard data-processing inequality; Linear SDPI refers to~\cref{prop:contraction_coeff_upper_bound}; and Non-Linear SDPI refers to~\cref{thm:non_linear_HS_div}. In this example setting, we consider $\gamma=6, \gamma'=2.5, \delta=0.01$ for a channel $\cN \in \cB^{\gamma,\delta}$. Each of these lines/curves show the largest $E_{\gamma'}\!\left(\cN(\rho) \Vert \cN(\sigma)\right)$ value that can be reached for the input distinguishability $E_{\gamma'}(\rho \Vert \sigma) \in [0,1]$. }
    \label{fig:compare}
\end{figure}

\item Next, we define $F_\gamma$ curves as a tool to obtain non-linear strong data-processing for divergences, with particular focus on hockey-stick divergence (see~\eqref{eq:f_gamma_def}). Then, we show that an upper bound on that can be obtained via~\cref{thm:non_linear_HS_div} so as to get~\cref{prop:F_gamma_1}. It is also interesting to know how $F_\gamma$ behaves when a composition of several channels is applied, so that one can utilize them in applications where the channel evolves at different timesteps. To this end, we obtain~\cref{prop:F_gamma_hetero_comp} for $\delta=0$, which is applicable for the setting where several channels (possibly heterogeneous) are applied sequentially one after the other. We also extend this result for the setting where $\delta >0$ in~\cref{prop:F_gamma_n_delta}. There, we find an interesting transition depending on the number of channel repetitions and the rate of contraction of the channel.

\item We study conditions that ensure a channel satisfies~\eqref{eq:intro_class_of_channels} for some $\gamma$ and $\delta$. This is important as we establish strong data-processing inequalities for hockey-stick divergences when the channel of interest satisfies~\eqref{eq:intro_class_of_channels}.

\item 
We establish reverse Pinsker-type inequalities for $f$-divergences in terms of hockey-stick divergence in~\cref{thm:revPinGeneral} and~\cref{Prop:RevPinCase1}. In previous works, reverse Pinsker inequalities give bounds on divergences in terms of the trace distance. Here, we explore the setting where we additionally have a constraint on certain hockey-stick divergences. 

\item We study how fast the hockey-stick divergence decays to a fixed value, which is what is referred to as mixing times in~\cref{prop:mixing_times1} and~\cref{prop:mixing_times_delta_not_0}. We in particular find cases where our non-linear SDPI can establish a finite time for the channel to reach its fixed point, which linear SDPI results cannot establish (\cref{rem:compare_mixing}). In addition, when the fixed point of the channel of interest is a full-rank state, we obtain possibly faster mixing times than the general setting even for the channels that do not satisfy the noise criterion in~\eqref{eq:intro_class_of_channels}, in~\cref{Cor:mixingtime2} and~\cref{Cor:mixingtime3}. 

\item Finally, we illustrate how our findings provide formal guarantees on the composition of private quantum systems in~\cref{prop:strong_privacy_eps_delta},~\cref{Cor:QLDP_2}, and~\cref{prop:delta_zero_mechanisms}, highlighting how one can generate stronger privacy mechanisms by composition. We also obtain contraction bounds on $f$-divergences under privacy constraints that hold even for $\delta \neq 0$, generalizing the previously known bounds for $\delta =0$ in both classical and quantum settings.
\end{itemize}

\textit{Note: Independent and concurrent work~ \cite{dasgupta2025quantum} also obtained a linear contraction coefficient for the hockey-stick divergence as in \cref{prop:contraction_coeff_upper_bound}, which is one ingredient to our study of non-linear SDPI, via a different proof method.} 

\section{Notations and Preliminaries}

 A quantum system $R$ is identified with a finite-dimensional Hilbert space~$\cH_R$. We denote the set of linear operators acting on $\cH_R$ by $\cL(\cH_R)$. A quantum state $\rho_R$ on $R$ is a positive semidefinite, unit-trace operator acting on $\cH_R$. We denote the set of all quantum states in 
$\cL(\cH_R)$ by $\cD(\cH_R)$. 
A state $\rho_R$ of rank one is called pure, and we may choose a normalized vector $| \psi \rangle \in \cH_R$ satisfying $\rho_R= | \psi \rangle\!\langle \psi | $ in this case. Otherwise,
$\rho_R$ is called a mixed state. 
A quantum channel $\cN \colon \cL(\cH_A ) \to \cL(\cH_B)$ is a linear, 
completely positive, and trace-preserving (CPTP) map from $\cL(\cH_A)$ to $\cL(\cH_B)$. We denote the composition of a channel $\cN_{A \to A}$ with itself $n$ times by $\cN^{(n)}$. A measurement of a quantum system $R$ is described by a
positive operator-valued measure (POVM) $\{M_y\}_{y \in \cY}$, which is defined to be a collection of PSD operators  satisfying $\sum_{y \in \cY} M_y= I_{R}$, where $\cY$ is a finite alphabet and $I_R$ is the identity operator on $R$. According to the Born rule, after applying the above POVM to $\rho \in \cD(\cH_R)$, the probability of observing the outcome $y$ is given by~$\Tr\!\left[M_y \rho \right]$.

 \bigskip
The normalized trace distance between the states $\rho$ and $\sigma$ is defined as
\begin{equation}
\label{eq:normalized-TD}
    T(\rho,\sigma)\coloneqq \frac{1}{2} \left\| \rho -\sigma\right\|_1.
\end{equation} 
It generalizes the total-variation distance between two probability distributions. For $\gamma \geq 0$,
the quantum hockey-stick divergence is defined as~\cite{sharma2012strongconversesquantumchannel}
\begin{equation}\label{eq:hockey_stick}
    E_\gamma(\rho \Vert \sigma) \coloneqq \Tr\!\left[(\rho -\gamma \sigma)_{+} \right] -(1-\gamma)_+,
\end{equation}
where $\left(  A\right)  _{+}\coloneqq \sum_{i:a_{i}\geq0}a_{i}|i\rangle\!\langle i| $
for a Hermitian operator $A = \sum_{i}a_{i}|i\rangle\!\langle i|$ and $(x)_+ \coloneqq \max\{0,x\}$ for scalars. For $\gamma=1$, observe that $ E_1(\rho \Vert \sigma)=T(\rho,\sigma)$. This shows the hockey stick divergence is a generalization of the trace distance. The  Hockey stick divergence is also known to enjoy the following semidefinite program formulation for $\gamma \geq 1$ \cite[Lemma II.1]{hirche2022quantum}
\begin{equation} \label{eq:HS_SDP}
    E_\gamma\!\left( \rho \Vert \sigma \right) = \sup_{0 \leq M \leq I} \Tr\!\left[ M(\rho-\gamma \sigma)\right] .
 \end{equation}
The max-relative entropy is defined as follows:
\begin{equation}
    D_{\max}(\rho\|\sigma)  \coloneqq \ln \inf_{\lambda \geq 0} \left\{ \lambda: \rho \leq \lambda \sigma\right\}. \label{eq:dmax-1-herm-maps}
\end{equation}
Finally, we will occasionally use a smoothed variant of the max-relative entropy, defined as
\begin{align}\label{Eq:smooth-Dmax}
    D^\delta_{\max}(\rho\|\sigma) &\coloneqq \log\inf\{ \lambda : Z\leq\lambda\sigma, Z\geq0, \Tr(\rho-Z)_+\leq\epsilon \} \\
    &= \log\inf\{ \lambda : \rho\leq\lambda\sigma+Q, Q\geq0, \Tr Q\leq\epsilon \} ,
\end{align}
where the equality between the two expressions was recently formally proven in~\cite{regula2025tight}. 

\subsection{Channels and Contraction Coefficients}

We say $\cN \in \cB^{\gamma, \delta}$ for $\gamma \geq 1$  and $\delta \in [0,1]$ if 
\begin{equation}\label{eq:gamma-delta-ball-def}
    \sup_{\rho, \sigma \in \cD(\cH)} E_{\gamma}\!\left( \cN(\rho) \Vert \cN(\sigma) \right) \leq \delta.
\end{equation}
In particular, $\cN \in \cB^{\gamma,\delta}$ guarantees that no two inputs are more distinguishable than $\delta$ under the hockey stick divergence after the application of the channel $\cN$. In this sense, $\cN \in \cB^{\gamma,\delta}$ is a convenient way to guarantee how much information is lost by an action of a channel $\cN$ without specifying the channel itself. This exact method of only specifying information loss is also used in differential privacy, see~\cite{hirche2022quantum}, and, it is readily seen from \cref{def:QLDP} in~\cref{Sec:privacy} that $\cN \in \cB^{\gamma,\delta}$ is equivalent to $\cN$ being $(\ln(\gamma),\delta)$-quantum local differentially private.
One example of such a channel is a depolarizing channel with the flip parameter $p$ such that $\gamma,\delta$ satisfy the following: $p \geq d(1-\delta)/(\gamma -1+d)$, where $d$ is the dimension \cite[ Lemma IV.1]{hirche2022quantum}. We discuss further conditions to achieve this condition in~\cref{Sec:Containment_class_of_channels}. 

Define the contraction coefficient for channel $\cN$ with respect to $E_\gamma$ as follows by choosing $\sD=E_\gamma$ in~\eqref{eq:SDPI_linear_opt}:
\begin{equation}\label{eq:eta-gamma-defn}
    \eta_{\gamma}(\cN) \coloneqq \sup_{\rho \neq \sigma} \frac{E_\gamma(\cN(\rho) \Vert \cN(\sigma))}{E_\gamma(\rho \Vert \sigma)}.
\end{equation}
Equivalently, we have that \cite{hirche2022quantum}, 
\begin{equation}\label{eq:contraction-optimized-by-pure-states}
     \eta_{\gamma}(\cN) =  \sup_{\psi_1 \perp \psi_2} E_\gamma( \cN(\psi_1) \Vert \cN(\psi_2)),
\end{equation}
  where the optimization is over $\psi_1$ and $\psi_2$ that are orthogonal pure states. {Moreover,  \cite[Proposition 11]{nuradha2024contraction} establishes 
  \begin{align}\label{eq:max-HS-dist-optimized-by-pure-states}
    \sup_{\rho, \sigma \in \cD(\cH)} E_{\gamma}\!\left( \cN(\rho) \Vert \cN(\sigma) \right) = \sup_{\psi_1 \perp \psi_2} E_\gamma( \cN(\psi_1) \Vert \cN(\psi_2)) \ . 
  \end{align}
   Combining \cref{eq:gamma-delta-ball-def}, \cref{eq:contraction-optimized-by-pure-states}, and \cref{eq:max-HS-dist-optimized-by-pure-states},} we have that 
  \begin{equation}\label{eq:ball-containment-equiv-contraction-coeff}
      \cN \in \cB^{\gamma, \delta}  \iff  \eta_{\gamma}(\cN) \leq \delta.
  \end{equation}

\section{Strong Data-Processing}
In this section, we establish strong data-processing inequalities (SDPI) for hockey-stick divergence. In particular, we establish both linear and non-linear SDPI for hockey-stick divergences that are tighter than the existing linear SDPI.

We begin by reviewing known results on the contraction bounds of hockey-stick divergence from previous works, so as to highlight the instances where one can improve linear SDPI and for comparisons to the SDPI obtained in this work: 
\begin{enumerate}
    \item \textit{Contraction Coefficients for $E_\gamma$:} Let $\cN \in \cB^{\gamma,\delta}$. We have that for $\gamma' \geq 1$ and $\delta=0$:
    \begin{equation} \label{eq:HS_Con_delta_0}
        E_{\gamma'}\!\left( \cN(\rho) \Vert \cN(\sigma) \right) \leq \left(\frac{\gamma -\gamma'}{\gamma +1} \right)_+ E_{\gamma'}(\rho \Vert \sigma).
    \end{equation}
    The above can be deduced from \cite[Theorem~1]{nuradha2024contraction}, and we can also recover it from \cref{prop:contraction_coeff_upper_bound} proved in this work with $\delta=0$. 
    Equation~\eqref{eq:HS_Con_delta_0} also says that for $\cN \in \cB^{\gamma,\delta=0}$, we get
    \begin{equation} \label{eq:contraction_hockey_delta_0}
        \eta_{\gamma'}(\cN) \leq \left(\frac{\gamma -\gamma'}{\gamma +1} \right)_+ .
    \end{equation}
    However, it is not clear whether this inequality is tight or can be improved. We show in this work that, in fact, it can be improved, and the tight characterization is, in contrast, a non-linear bound. 
    
     \item \textit{Contraction Coefficients for $E_1$ (trace distance) :} Let $\cN \in \cB^{\gamma,\delta}$. We have that for normalized trace distance $(\gamma'=1)$ \cite[Theorem~5]{nuradha2024contraction} (see also~\cite[Proposition~4.5]{Christoph2024sample} when $\delta=0$)
    \begin{equation}\label{Eq:E1-contraction}
        E_{1}\!\left( \cN(\rho) \Vert \cN(\sigma) \right) \leq \left(\frac{\gamma -1 + 2 \delta }{\gamma +1} \right) E_{1}(\rho \Vert \sigma).
    \end{equation}
    Also, there exists a channel $\cN$ that achieves the above inequality (by the achievability proof of \cite[Theorem~5]{nuradha2024contraction}). Therefore, the linear contraction bound is in fact tight for $\gamma=1$ .
\end{enumerate}

We now turn to establishing our non-linear SDPI (\cref{thm:non_linear_HS_div}). We first derive relations between hockey-stick divergences of different parameter values. These are quantum generalizations of~\cite[Propositions 1 and 2]{zamanlooy2024mathrm}, which were key technical lemmas in that work for establishing non-linear SDPI for classical hockey-stick divergences. We begin with a quantum generalization of \cite[Proposition 1]{zamanlooy2024mathrm}.

\begin{proposition} \label{prop:HS_relationship_gamma'}
    Let $1 \leq \gamma' \leq \gamma$. Then, we have that 
    \begin{equation}
        E_{\gamma'}(\rho \Vert \sigma) \leq \frac{\gamma- \gamma'}{\gamma+1} + \frac{\gamma'+1}{\gamma+1} \max\{ E_\gamma(\rho \Vert \sigma), E_\gamma(\sigma \Vert \rho)\}.
    \end{equation}
\end{proposition}

\begin{proof} Let 
\begin{equation}
    t \coloneqq \max\{ E_\gamma(\rho \Vert \sigma), E_\gamma(\sigma \Vert \rho)\}.
\end{equation}
and denote 
\begin{equation}
    \delta' \coloneqq \frac{\gamma- \gamma'}{\gamma+1} + \frac{\gamma'+1}{\gamma+1} t.
\end{equation}
Using the SDP formulation of $E_{\gamma}$ and the constraint $0\leq M \leq I$, we have that 
\begin{align}
    \Tr[ M \rho] - \gamma \Tr[ M \sigma] \leq t.
\end{align}
That leads to the following:
\begin{align}
    \Tr[ M \rho] - \gamma' \Tr[ M \sigma]&  \leq \gamma \Tr[ M \sigma] +t - \gamma' \Tr[ M \sigma] \\
    & \leq \delta' + (\gamma -\gamma') \left( \Tr[ M \sigma]- \frac{(1-t)}{\gamma +1} \right) \label{eq:first_M}.
\end{align}
Also, since $0 \leq I- M \leq I$, we also have that 
\begin{equation}
    \Tr[ (I-M) \sigma] - \gamma \Tr[(I- M) \rho] \leq t .
\end{equation}
With that, 
\begin{align}
    \Tr[M \rho] \leq 1- \gamma^{-1}(1-t)+ \gamma^{-1} \Tr[ M \sigma].
\end{align}
Then again with similar algebraic manipulations, we get
\begin{align}
    \Tr[M \rho] -\gamma' \Tr[M \sigma]& \leq  -\gamma' \Tr[M \sigma] + 1- \gamma^{-1}(1-t)+ \gamma^{-1} \Tr[ M \sigma] \\
    & = \delta' +(\gamma - (\gamma')^{-1}) \left( -\Tr[M \sigma]  +\frac{1-t}{\gamma +1}\right). \label{eq:second_M}
\end{align}
With~\eqref{eq:first_M}, if $\Tr[ M \sigma] \leq  (1-t)/(\gamma +1)$, we get that $\Tr[M \rho] -\gamma' \Tr[M \sigma] \leq \delta'$.
Also by~\eqref{eq:second_M}, if $\Tr[ M \sigma] > (1-t)/(\gamma +1)$, $ \Tr[M \rho] -\gamma' \Tr[M \sigma] \leq \delta'$. 
Therefore, for all $0 \leq M \leq I$, we have that 
\begin{equation}
    \Tr[M \rho] -\gamma' \Tr[M \sigma] \leq \delta'.
\end{equation}
By recalling,~\eqref{eq:HS_SDP} again, we have that 
\begin{equation}
    E_{\gamma'}(\rho \Vert \sigma) \leq \delta',
\end{equation}
concluding the proof. 
\end{proof}

We now generalize~\cite[Proposition 2]{zamanlooy2024mathrm}. We note this result will also be useful in obtaining faster mixing times and stronger privacy composition guarantees in~\cref{Sec:Mixing_times} and \cref{Sec:privacy}, respectively.
 \begin{proposition}
 \label{Prop:gamma-gamma-ev}
     Set $\gamma\geq\gamma'\geq1$. If we have
     \begin{align}
         E_{\gamma'}(\rho\|\sigma) \leq (\gamma-\gamma')\, \lambda_{\min}(\sigma), 
     \end{align}
     then
     \begin{align}
         E_\gamma(\rho\|\sigma) = 0. 
     \end{align}
 \end{proposition}
\begin{proof}
    First, note that if $\sigma$ has $\lambda_{\min}(\sigma)=0$ then the above is trivially true by monotonicity in $\gamma$. Hence, we only consider full rank $\sigma$ in the following. 

    According to~\cite{hirche2022quantum} small $E_\gamma$ can be related to a modified smooth $D_{\max}$ divergence, as defined in~\cref{Eq:smooth-Dmax}, which is equivalent to the information spectrum divergence, à la Datta-Leditzky, by e.g.~\cite{regula2025tight}. That implies that if $E_{\gamma'}(\rho\|\sigma) \leq\epsilon$ there exists an operator $Q\geq0$ with $\tr Q\leq \epsilon$ such that $\rho\leq\gamma'\sigma+Q$. Note that we have then $Q \leq \epsilon\Id \leq \epsilon\sigma/\lambda_{\min}(\sigma)$.
    This gives, 
    \begin{align}
        \tr(\rho-\gamma\sigma)_+ \leq \tr(\gamma'\sigma+Q-\gamma\sigma)_+ \\
        \leq \tr((\gamma'-\gamma)\sigma +\frac{\epsilon}{\lambda_{\min}(\sigma)}\sigma )_+\,,
    \end{align} 
    which implies the claim by setting $\epsilon=(\gamma-\gamma')\, \lambda_{\min}(\sigma)$. 
\end{proof}

Based on the above we state a brief corollary that might be of independent interest. It gives an upper bound on the $D_{\max}$ divergence. 
\begin{corollary}
    For $\gamma\geq1$, 
    \begin{align}
    D_{\max}(\rho\|\sigma) \leq \log\left( \gamma + \frac{E_{\gamma}(\rho\|\sigma)}{\lambda_{\min}(\sigma)} \right).
\end{align}
\end{corollary}
\begin{proof}
    Fix some $\gamma'\geq1$. Say $E_{\gamma'}(\rho\|\sigma) = x$. Then there exists a $\gamma\geq\gamma'$ such that $x=(\gamma-\gamma')\, \lambda_{\min}(\sigma)$. It follows that $E_\gamma(\rho\|\sigma)=0$ by~\cref{Prop:gamma-gamma-ev}. We also have $\gamma = \frac{x}{\lambda_{\min}(\sigma)}+\gamma'$. However, we also know that the first point at which the Hockey-Stick divergence becomes zero is given by the $D_{\max}$ divergence, hence the above gives an upper bound as stated in the claim. 
\end{proof}
The special case of $\gamma=1$ for classical probability distributions can be found in~\cite{sason2015reverse}. 
As a direct consequence, one can also prove the following statement, which will be useful later in obtaining contraction bounds on general families of divergences in~\cref{prp:f_div_Contraction}.

\begin{corollary}\label{Cor:Dmax-by-smooth}
    For $0\leq\delta\leq1$, we have
    \begin{align}
    D_{\max}(\rho\|\sigma) \leq \log\left( e^{D^\delta_{\max}(\rho\|\sigma)} + \frac{\delta}{\lambda_{\min}(\sigma)} \right).
\end{align}
\end{corollary}
\begin{proof}
    This follows from the duality between the smooth max-divergence and the Hockey-Stick divergence. 
\end{proof}

Next, we present a linear SDPI for hockey-stick divergence. With that we also recover~\eqref{eq:contraction_hockey_delta_0} from~\cite{nuradha2024contraction} for $\delta =0$.
\begin{proposition} \label{prop:contraction_coeff_upper_bound}
    Let $\cN \in \cB^{\gamma, \delta}$ and $\gamma \geq \gamma' \geq 1$. Then, we have that 
    \begin{equation} \label{eq:contraction_HS_numerator}
        E_{\gamma'}\!\left( \cN(\rho) \Vert \cN(\sigma) \right) \leq \frac{(\gamma- \gamma') +\delta (\gamma'+1) }{\gamma+1} .
    \end{equation}
    Furthermore, we have that for $\gamma \geq \gamma' \geq 1$
    \begin{equation}\label{eq:bound-on-HS-contraction-coeff}
        \eta_{\gamma'}(\cN) \leq \frac{(\gamma- \gamma') + \delta (\gamma'+1) }{\gamma+1} .
    \end{equation}
For $\gamma' \geq \gamma$, we have that 
 \begin{equation}
        \eta_{\gamma'}(\cN) \leq \delta.
    \end{equation}
This results in for $\gamma' \geq 1$, 
\begin{equation}
        \eta_{\gamma'}(\cN) \leq \max \left\{  \frac{(\gamma- \gamma') + \delta (\gamma'+1) }{\gamma+1} , \delta \right\}.
    \end{equation}
\end{proposition}

\begin{proof}
    Since $\cN \in \cB^{\varepsilon, \delta}$, we have that 
    \begin{equation}
        \max\left \{ E_\gamma(\cN(\rho) \Vert \cN(\sigma)),  E_\gamma(\cN(\sigma) \Vert \cN(\rho))\right\} \leq \delta.
    \end{equation}
    Then, by applying~\cref{prop:HS_relationship_gamma'}, we have that 
    \begin{equation}
        E_{\gamma'}\!\left( \cN(\rho) \Vert \cN(\sigma) \right) \leq \frac{(\gamma- \gamma') +\delta (\gamma'+1) }{\gamma+1}.
    \end{equation}
    For the contraction coefficient, recall that
    \begin{equation}
        \eta_{\gamma'}(\cN) = \sup_{\psi_1 \perp \psi_2} E_{\gamma'}( \cN(\psi_1) \Vert \cN(\psi_2)),
    \end{equation}
    where the optimization is over $\psi_1$ and $\psi_2$ that are orthogonal pure states. This immediately concludes the proof by applying the first proposition statement for $\gamma \geq \gamma' \geq 1$. For $\gamma' \geq \gamma$, it follows as above together with the monotonicity of hockey-stick divergence.

    For the last inequality, consider the two inequalities that we proved for those two separate regimes. See that 
    \begin{equation}
         \frac{(\gamma- \gamma') + \delta (\gamma'+1) }{\gamma+1}  - \delta = \frac{(1-\delta) (\gamma -\gamma')}{(\gamma +1)}. 
    \end{equation}
    From this, it is evident that the first bound is strictly larger than $\delta >0$ if and only if $\gamma > \gamma'$. This leads to the combination of both the regimes, concluding the proof.
\end{proof}

\begin{remark}{Measured Hockey-stick Divergences}\label{rem:measured_HS_SDPI}
Measured hockey-stick divergences are defined as follows \cite{nuradha2025MeasuredHS}: for $\gamma \geq 1$
\begin{equation}\label{eq:hockey_stick_define}
    E_\gamma^{\mathcal{M}}(\rho \Vert \sigma) \coloneqq \sup_{M\in \mathcal{M}_2} \left\{\operatorname{Tr}\!\left[ M(\rho -\gamma \sigma )\right]\right\}, 
\end{equation}
where $\mathcal{M}$ denotes the restricted measurement operator set with
\begin{equation}\label{eq:M_2_set}
    \cM_2 \coloneqq \left\{ M: M, I-M \in \cM \right\}.
\end{equation}
Let $\cN \in \cB^{\gamma,\delta}_{\cM,\cS}$ if for $\mathcal{S} \subseteq \cD(\cH)$ 
\begin{equation}
    \sup_{\rho, \sigma \in \cS} E_{\gamma}^\cM \!\left( \cN(\rho) \Vert \cN(\sigma) \right) \leq \delta.
\end{equation}
Also, note that $\cN \in \cB^{\gamma,\delta} \implies \cN \in \cB^{\gamma,\delta}_\cM$. Thus, $\cN \in \cB^{\gamma,\delta}_{\cM,\cS}$ can be understood as a relaxed criterion or a criterion that takes into account the measurements that can be applied in the relevant practical setting (e.g.; local operations and classical communications in contrast to joint measurements over all sub-systems) and a subset of states $\cS$ that is relevant for that application and the experiment.
Considering this relaxed setting, \cref{prop:HS_relationship_gamma'} and~\eqref{eq:contraction_HS_numerator} can be extended with similar proof arguments by considering measurements in the set $\cM$ as follows: let $1 \leq \gamma' \leq \gamma$. Then, we have that 
    \begin{equation}
        E_{\gamma'}^\cM (\rho \Vert \sigma) \leq \frac{\gamma- \gamma'}{\gamma+1} + \frac{\gamma'+1}{\gamma+1} \max\{ E_\gamma^\cM (\rho \Vert \sigma), E_\gamma^\cM (\sigma \Vert \rho)\}.
    \end{equation}
Also,  let $\cN \in \cB^{\gamma, \delta}_{\cM, \cS}$ and $\gamma \geq \gamma' \geq 1$. Then, we have that 
    \begin{equation} \label{eq:contraction_HS_numerator_M_specific}
        E_{\gamma'}^\cM\!\left( \cN(\rho) \Vert \cN(\sigma) \right) \leq \frac{(\gamma- \gamma') +\delta (\gamma'+1) }{\gamma+1} .
    \end{equation}
For $\gamma'=1$, the above claim is also given in~\cite[Proposition~8]{nuradha2025MeasuredHS} for $\gamma=e^\varepsilon$.
\end{remark}

We present another key ingredient in obtaining non-linear SDPI by improving linear SDPI for hockey-stick divergences.

\begin{lemma} \label{lem:contraction_separate_gamma}
    We have that for $\gamma \geq 1$
     \begin{align}
        E_{\gamma}(\cN(\rho)\|\cN(\sigma)) \leq \eta_{\beta}(\cN)\, E_{\gamma}(\rho\|\sigma),
    \end{align}
    where $\beta \coloneqq 1-\frac{1-\gamma}{E_\gamma(\rho\|\sigma)}$.
\end{lemma}
\begin{proof}
Recall that 
\begin{align}
    \rho-\gamma\sigma &= X_+ - X_- \\
    &= \lambda_+ \tau_+ - \lambda_- \tau_-, 
\end{align}
where $\lambda_+=\Tr X_+ = E_\gamma(\rho\|\sigma)$, $\lambda_-=\Tr X_-$, $\tau_\pm=\frac{X_\pm}{\lambda_\pm}$.  
The crucial step here is the following observation: 
\begin{align}
    &1-\gamma = \lambda_+ - \lambda_- \Leftrightarrow \frac{\lambda_-}{\lambda_+} = 1 - \frac{1-\gamma}{\lambda_+}. 
\end{align}
Consider,
    \begin{align}
        E_\gamma(\cN(\rho)\|\cN(\sigma)) &= \lambda_+ E_{\frac{\lambda_-}{\lambda_+}}(\cN(\tau_+)\|\cN(\tau_-)) \\ 
        &=\lambda_+ E_{1 - \frac{1-\gamma}{\lambda_+}}(\cN(\tau_+)\|\cN(\tau_-)) \\ 
        &=E_\gamma(\rho\|\sigma) E_{1 - \frac{1-\gamma}{E_\gamma(\rho\|\sigma)}}(\cN(\tau_+)\|\cN(\tau_-)) \\ 
        &\leq E_\gamma(\rho\|\sigma) \eta_{1-\frac{1-\gamma}{E_\gamma(\rho\|\sigma)}}(\cN), 
    \end{align}
    thus concluding the proof.
\end{proof}

Now, with the above results, we establish a non-linear SDPI for hockey-stick divergences. 

\begin{theorem} \label{thm:non_linear_HS_div}
    Let $\cN \in \cB^{\gamma, \delta}$ with $\gamma \geq 1$ and $\delta \in [0,1]$. Then, we have that for $\gamma' \geq 1$
    \begin{equation}
        E_{\gamma'}\!\left( \cN(\rho) \Vert \cN(\sigma) \right) \leq \max \left\{\frac{(\gamma +2 \delta -1) E_{\gamma'}(\rho \Vert \sigma) -(\gamma'-1)(1-\delta)}{\gamma +1}, \delta E_{\gamma'}(\rho \Vert \sigma) \right\}.
    \end{equation}

\end{theorem}
\begin{proof}
First, by~\cref{lem:contraction_separate_gamma}, we have that for $1 \leq \beta= 1- (1-\gamma') /E_{\gamma'}(\rho \Vert \sigma)  $
\begin{align}
      E_{\gamma'}( \cN(\rho) \Vert \cN(\sigma)) &\leq \eta_{\beta}(\cN)\, E_{\gamma'}(\rho\|\sigma) \\
      & \leq  \max \! \left\{\frac{(\gamma- \beta) + \delta (\beta+1) }{\gamma+1}, \delta \right\} E_{\gamma'}(\rho\|\sigma),
\end{align}
where the last inequality follows from~\cref{prop:contraction_coeff_upper_bound}. We arrive at the desired inequality by algebraic simplifications together with the substitution of $\beta= 1-  (1-\gamma')/E_{\gamma'}(\rho \Vert \sigma)$. 
\end{proof}
\begin{remark}[Achievability of the Non-Linear SDPI]
    We show that the upper bound we get there is in fact tight in several cases, meaning that there exists a channel $\cN \in \cB^{\gamma, \delta}$ that achieves the equality. 
To this end, consider $\cA = \cA_{\operatorname{Dep}}^p \circ \cM$ such that the measurement channel $\cM_{\rho,\sigma}$ is defined as follows: by choosing $M$ such that it is the projection to the positive eigenspace of $\rho -\gamma' \sigma$
\begin{equation}
    \cM(\omega) = \Tr[M \omega] |0\rangle\!\langle 0| + (1- \Tr[M \omega]) |1\rangle\!\langle 1|. 
\end{equation} This leads to
\begin{align}
    & \cM(\rho)= E_{\gamma'}(\rho \Vert \sigma) |0\rangle\!\langle 0| + (1- E_{\gamma'}(\rho \Vert \sigma)) |1\rangle\!\langle 1| \\
    &\cM(\sigma)=(1- E_{\gamma'}(\rho \Vert \sigma)) |0\rangle\!\langle 0| + E_{\gamma'}(\rho \Vert \sigma) |1\rangle\!\langle 1|.
\end{align}

We also have that if $p = 2(1- \delta)/(\gamma +1)$, then $\cA = \cA_{\operatorname{Dep}}^p \circ \cM$, we have $\eta_{\gamma}(\cA) \leq \delta$ by applying~\cite[Lemma IV.1]{hirche2022quantum} as done in~\cite[Eq (186)-(189)]{nuradha2024contraction}.

We also have that 
\begin{equation}
    \cA(\omega) = \left(\Tr[M \omega] (1-p) +\frac{p}{2} \right) | 0\rangle\!\langle 0| + \left(\Tr[M \omega] (p-1) + 1-\frac{p}{2} \right) | 1\rangle\!\langle 1|. 
\end{equation}
This leads to 
\begin{align}
    & E_{\gamma'}\!\left( \cA(\rho) \Vert \cA(\sigma) \right) \notag \\
    &= \left( \Tr[M(\rho-\gamma' \sigma)] (1-p) +\frac{p}{2} (1-\gamma') \right)_+ + \left( \Tr[M(\rho-\gamma' \sigma)] (p-1) + \left(1-\frac{p}{2}\right) (1-\gamma') \right)_+.
\end{align}
Since $\gamma' \geq 1$ and $p =2(1-\delta)/(\gamma +1)\in [0,1]$ for $\gamma \geq 1$ and $\delta \in [0,1]$, and recalling that  $M$ is the projection onto the positive eigenspace of $\rho -\gamma' \sigma$, we observe that only the first term survives as follows:
\begin{align}
    E_{\gamma'}\!\left( \cA(\rho) \Vert \cA(\sigma) \right)  & =\left( \Tr[M(\rho-\gamma' \sigma)] (1-p) +\frac{p}{2} (1-\gamma') \right)_+ \\
    &= \left( E_{\gamma'}(\rho \Vert \sigma) (1-p) +\frac{p}{2} (1-\gamma') \right)_+ \\ 
    &= \left( E_{\gamma'}(\rho \Vert \sigma) \frac{(\gamma -1 +2 \delta)}{(\gamma +1)} +\frac{(1-\delta)}{(\gamma +1)} (1-\gamma') \right)_+,
\end{align}
where the last equality follows by substituting $p=2(1-\delta)/(\gamma +1)$.

With this, we see that the equality is achieved when $\delta=0$. 
Furthermore, if $1 + (\gamma' -1)/ E_{\gamma'}(\rho \Vert \sigma) \leq \gamma$, the second term inside the maximum in~\cref{thm:non_linear_HS_div} doesn't get activated. This leads to the case that for all $\delta\in [0,1]$ and 
 $E_{\gamma'}(\rho \Vert \sigma) \geq (\gamma' -1) /(\gamma +1)$ also, the equality holds.
\end{remark}

\section{$F_\gamma$ Curves}

In this section, we define $F_\gamma$ curves for hockey-stick divergence by having $\sD=E_\gamma$ in~\eqref{eq:divergence_curve} and obtain non-linear SDPI for settings where we analyse composite channel formed by sequential composition of several channels.

Let us define $F_\gamma$ curves as follows: For $t \in [0,1]$

\begin{equation} \label{eq:f_gamma_def}
    F_\gamma^\cN(t) \coloneqq \sup_{\rho,\sigma} \left\{ E_\gamma\!\left( \cN(\rho) \Vert \cN(\sigma) \right): E_\gamma (\rho \Vert \sigma) \leq t \right\}.
\end{equation}
Note that for $\gamma=1$, it reduces to Dobrushin curve in~\cite{huber2019jointly}. Also, $F_\gamma$ satisfies the following properties:
\begin{enumerate}
    \item For $0 \leq t_1 \leq t_2 \leq 1$, we have that 
    \begin{equation} \label{eq:monotonicity_F_g}
        F_\gamma^\cN(t_1) \leq F_\gamma^\cN(t_2).
    \end{equation}
    \item By the above fact and data-processing of $E_\gamma$ and for two states $\rho,\sigma$ such that $E_\gamma(\rho \Vert \sigma) \leq t$, we have 
    \begin{equation}
      E_{\gamma}  (\cN_{2} \circ \cN_{1} (\rho) \Vert \cN_{2} \circ \cN_{1}(\sigma))  \leq F_{\gamma}^{\cN_{2}}(E_\gamma(\cN_{1}(\rho) \Vert \cN_{1}(\sigma))) \leq F_{\gamma}^{\cN_{2}}(F_\gamma^{\cN_{1}}(t)),
    \end{equation}
      which leads to 
    \begin{equation}
        F_\gamma^{\cN_{2} \circ \cN_{1}}(t) \leq F_\gamma^{\cN_{2}}(F_\gamma^{\cN_{1}}(t)). \label{eq:composition_property}
    \end{equation}
\end{enumerate}
In the following we discuss upper bounds on this curve. 
\begin{proposition} \label{prop:F_gamma_1} Let $\cN \in \cB^{\gamma, \delta}$. Then, we have that for $\gamma' \geq 1$
\begin{equation}
    F_{\gamma'}^\cN(t) \leq  \max \left\{\frac{(\gamma +2 \delta -1) t -(\gamma'-1)(1-\delta)}{\gamma +1}, \delta t\right\}.
\end{equation}
\end{proposition}
\begin{proof}
    The proof follows by applying~\cref{thm:non_linear_HS_div} together with the definition of $F_\gamma$ in~\eqref{eq:f_gamma_def}.
\end{proof}

For the remainder of the section, we will use the above to bound the $F_\gamma$ curve under sequential composition of quantum channels. 
\begin{proposition}[Sequential Composition of Heterogeneous Channels] \label{prop:F_gamma_hetero_comp}
    Let $\cN_i \in \cB^{\gamma_i,\delta=0}$ for $\gamma_i \geq 1$ and $1 \leq \gamma' \leq \gamma_i$ with all $i \in \{1, \ldots,n\}$. Denote $\cM \coloneqq \cN_n \circ \cdots \circ \cN_1 $. For $t \in[0,1]$:
    \begin{equation}
        F_{\gamma'}^{\cM}(t) \leq \left( t \prod_{i=1}^n \frac{\gamma_i-1}{\gamma_i +1}  -\frac{(\gamma'-1)}{2} \left(1-  \prod_{i=1}^n \frac{\gamma_i-1}{\gamma_i +1}\right)  \right)_+.
    \end{equation}
\end{proposition}
\begin{proof}
    We prove this by induction. First, for $N=1$, we have from~\cref{prop:F_gamma_1} that 
    \begin{align}
        F_{\gamma'}^{\cN_1}(t) & \leq \left( \frac{\gamma_1-1}{\gamma_1 +1} t - \frac{ (\gamma'-1)}{\gamma_1+1}\right)_+ \\
        &=\left( \frac{\gamma_1-1}{\gamma_1 +1} t - \frac{ (\gamma'-1)}{2} \left( 1- \frac{(\gamma_1-1)}{\gamma_1+1}\right)\right)_+,
    \end{align}
    which derives the desired inequality. 

    As the induction hypothesis, we have that for $N=n-1$
    \begin{equation}
         F_{\gamma'}^{\cN_{n-1} \circ \cdots \circ \cN_1}(t) \leq \left( t \prod_{i=1}^{n-1} \frac{\gamma_i-1}{\gamma_i +1}  -\frac{(\gamma'-1)}{2} \left(1-  \prod_{i=1}^{n-1} \frac{\gamma_i-1}{\gamma_i +1}\right)  \right)_+.
    \end{equation}
For $N=n$, due to the sequential composition property of $F_{\gamma'}$ in~\eqref{eq:composition_property} together with $N=1$ case, we have that 

\begin{align*}
     F_{\gamma'}^{\cM}(t) & \leq \left( \frac{\gamma_n-1}{\gamma_n +1} F_{\gamma'}^{\cN_{n-1} \circ \cdots \circ \cN_1}(t) - \frac{ (\gamma'-1)}{2} \left( 1- \frac{(\gamma_n-1)}{\gamma_n+1}\right)\right)_+ \\
     &\leq \Bigg[\frac{\gamma_{n}-1}{\gamma_{n}+1}\left( t \prod_{i=1}^{n-1} \frac{\gamma_i-1}{\gamma_i +1}  -\frac{(\gamma'-1)}{2} \left(1-  \prod_{i=1}^{n-1} \frac{\gamma_i-1}{\gamma_i +1}\right)  \right)_+ - \frac{ (\gamma'-1)}{2} \left( 1- \frac{(\gamma_n-1)}{\gamma_n+1}\right) \Bigg]_{+} \\ 
     &= \Bigg[\left( t \prod_{i=1}^{n} \frac{\gamma_i-1}{\gamma_i +1}  -\frac{(\gamma'-1)}{2}\frac{\gamma_{n}-1}{\gamma_{n}+1} \left(1-  \prod_{i=1}^{n-1} \frac{\gamma_i-1}{\gamma_i +1}\right)  \right)_+  - \frac{ (\gamma'-1)}{2} \left( 1- \frac{(\gamma_n-1)}{\gamma_n+1}\right) \Bigg]_{+}, 
\end{align*}
 where the second inequality follows from the induction hypothesis. 
Then, by the fact that $((a)_{+} -b)_{+} = (a-b)_{+}$ for $a \in \mbb{R}, b \geq 0$ together with $\frac{ (\gamma'-1)}{2} \left( 1- \frac{(\gamma_n-1)}{\gamma_n+1}\right) \geq 0$ by assumption on $\gamma_{n}$, we have 
\begin{align*}
     F_{\gamma'}^{\cM}(t) & \leq \Bigg[ t \prod_{i=1}^{n} \frac{\gamma_i-1}{\gamma_i +1}  -\frac{(\gamma'-1)}{2}\left( \frac{\gamma_{n}-1}{\gamma_{n}+1} \left(1-  \prod_{i=1}^{n-1} \frac{\gamma_i-1}{\gamma_i +1}\right) + \left( 1- \frac{(\gamma_n-1)}{\gamma_n+1}\right) \right) \Bigg]_{+} \\
     &= \Bigg[ t \prod_{i=1}^{n} \frac{\gamma_i-1}{\gamma_i +1}  -\frac{(\gamma'-1)}{2}\left(1-  \prod_{i=1}^{n} \frac{\gamma_i-1}{\gamma_i +1}\right) \Bigg]_{+} \ , 
\end{align*}
where the second equality follows by algebraic simplifications.
\end{proof}

\begin{corollary}
\label{prop:F_gamma_n_channels}
    Let $\cN \in \cB^{\gamma, \delta=0}$. Then, for the sequential composition of channel $\cN$ with $\cN^{(n)} \coloneqq \underbrace{\cN \circ \cdots \circ \cN}_{n}$. We have that for $\gamma' \geq 1$
    \begin{equation}
      F_{\gamma'}^{\cN^{(n)}}(t) \leq   \left( t \left( \frac{\gamma -1}{ \gamma +1} \right)^n  -\frac{(\gamma'-1)}{2} \left(1-  \left( \frac{\gamma -1}{ \gamma +1} \right)^n \right)  \right)_+ 
   \end{equation}
    
    and \begin{equation}
        F_{\gamma'}^{\cN^{(n)}}(t) \leq \frac{1}{2}\left( \left( \frac{\gamma -1}{ \gamma +1} \right)^n (\gamma'+1) +1 -\gamma'  \right)_+.
    \end{equation}    
\end{corollary}
\begin{proof}
    Proof follows by adapting~\cref{prop:F_gamma_hetero_comp} by $\cN_i =\cN$ for all $i \in\{1, \ldots, n\}$. The second inequality follows by the fact that $t \leq 1$ and~\eqref{eq:monotonicity_F_g}. 
\end{proof}

\bigskip 
The above characterizations of the $F_\gamma$ curves are for the setting $\delta=0$.  Utilizing~\cref{thm:non_linear_HS_div} for the setting $\delta \neq 0$, we have to carefully evaluate which term achieves the maximum therein. 

Define the following shorthand notations:
\begin{equation}
  a \coloneqq\ \frac{\gamma+2\delta-1}{\gamma+1} ,\qquad
  b\coloneqq \frac{(\gamma'-1)(1-\delta)}{\gamma+1}.
\end{equation}
Also define 
\begin{equation}
    t_* \coloneqq \frac{b}{a-\delta} = \frac{\gamma'-1}{\gamma -1} 
\end{equation}
and see that for $1 < \gamma' <\gamma$ and $\delta \in (0,1)$, we get $t_* \in (0,1)$.

 Set, for $k\in \mathbb{N}$, for $t \in [0,1]$
\begin{equation}
  \Phi_k(t)\coloneqq a^{k}\!\left(t+\frac{b}{1-a}\right) - \frac{b}{1-a}.
  \label{eq:phik}
\end{equation}
and define the first hitting time as
  \begin{align}
    k_*(t) &\coloneqq \min\{k\in\mathbb N:\ \Phi_k(t)\le t_*\}
    \\&=\left\lceil\frac{\ln\!\left(\frac{t_*(1-a)+b}{\,t(1-a)+b\,}\right)}{\ln a}\right\rceil_+.
    \label{eq:kstar}
  \end{align}

  \begin{proposition} \label{prop:F_gamma_n_delta}
      Let $1 < \gamma' <\gamma$, $\delta \in (0,1)$ and $\cN_i \in \cB^{\gamma, \delta}$ for $i\in\{1,\dots,n\}$. Denote $\cM \coloneqq \cN_n \circ \cdots \circ \cN_1 $. We have for $t \in[0,1]$ that 
      \begin{equation}
           F_{\gamma'}^\cM(t) \leq G_n(t),
      \end{equation}
    where  \begin{equation}
 G_n(t) \coloneqq
    \begin{cases}
      \Phi_n(t), & 1\le n\le k_*(t),\\
      \delta^{\,n-k_*(t)}\,\Phi_{k_*(t)}(t), & n > k_*(t).
    \end{cases}
    \label{eq:Gn-piecewise}
  \end{equation}
  \end{proposition}
\begin{proof}
    We prove this using an induction argument. 
    First for $n=1$, as $a \in (0,1)$ by the assumptions on the parameters, when $t \geq t_{\ast}$, we get $k_{\ast}(t) \geq 1$, so $G_{1}(t) = \Phi_{1}(t)$. By the same reasoning, when $t < t_{\ast}$, $k_{\ast}(t) = 0$, so that $G_{1}(t) = \delta \cdot \Phi_{0}(t) = \delta t$.
    By using~\cref{prop:F_gamma_1}, we confirm that for $n=1$, the required claim holds.

    As the induction hypothesis, we assume that the claim holds for $n=m$. With that, we prove that the claim holds for $n=m+1$ as follows: for $\cN_i=\cN$ for all $i$
    \begin{align}
        F_{\gamma'}^{\cN_{m+1} \circ \cdots \circ \cN_1}(t) & \leq F_{\gamma'}^{\cN_{m+1}}\!\left( F_{\gamma'}^{\cN_{m} \circ \cdots \circ \cN_1} (t)\right) \\ 
        & \leq F_{\gamma'}^{\cN_{m+1}}\!\left( G_m(t)\right) \\
        & \leq G_1\left( G_m (t)\right) \\
        &=  G_{m+1}(t),
    \end{align}
    where the first inequality from~\eqref{eq:composition_property}; the second inequality from the induction hypothesis and the monotonicity in~\eqref{eq:monotonicity_F_g}; third by the proved claim for an application of a single channel in $n=1$ case by substituting $ G_m(t) \leq t_*$ for $t$ therein; and the last equality follows due to the reasoning discussed below.

\medskip
    \textit{Proof of $G_1\left( G_m (t)\right)= G_{m+1}(t) $}: We consider three cases. First, consider $m < k_*(t).$ 
    Then $G_m(t) =\Phi_m(t)$. 
    Recall that 
    \begin{align}
    G_{1}(t) = \begin{cases} a(t+\frac{b}{1-a}) - \frac{b}{1-a} & k_{\ast}(t) \geq 1 \\
    \delta t & k_{\ast}(t) < 1 \ .
    \end{cases}
\end{align}
Since $m < k_*(t)$, we have $G_m(t) =\Phi_m(t) > t_*$. 
Then, by recalling that $ k_*(G_m(t)) = \min\{k\in\mathbb N:\ \Phi_k(G_m(t))\le t_*\}$, 
we get $k_*(G_m(t)) \geq 1$. 
With that, the first branch of the $G_1$ gets activated as follows: 
    \begin{align}
     G_1\left( G_m (t)\right) &=   G_1\left( \Phi_m (t)\right)  \\ 
     &= a \Phi_m (t) -b \label{eq:case-1-step} \\ 
     &= \Phi_{m+1}(t).
    \end{align}
Since $m+1 \leq  k_*(t)$, we have that $G_{m+1}(t)= \Phi_{m+1}(t)$ so as to get $G_1(G_m) = G_{m+1}$.

Second, consider $m=k_*(t)$. That leads to 
$G_m(t)= \Phi_{k_*(t)}(t) $. Due to $\Phi_{k_*(t)}(t) \leq t_*$, we have that $k_*(G_m(t))=0 <1$ activating the second branch of $G_1$ to get
\begin{equation}
    G_1\!\left(G_m(t)\right)
= G_1\!\left(\Phi_{k_*(t)}(t)\right)
= \delta \Phi_{k_*(t)}(t).
\end{equation}
Since $m= k_*(t)$, we also have that  $G_{m+1}(t)= \delta \Phi_{k_*(t)}(t)$.  for $m+1 > k_*(t)$. With that, we conclude $G_1(G_m) = G_{m+1}$ here as well.

Third, consider $m > k_*(t)$, and the second branch of $G_m(t)$ getting activated leading to 
\begin{equation}
    G_m(t)
= \delta^{\,m - k_*(t)}\,\Phi_{k_*(t)}(t)
\le t_*
\end{equation}
Now since $G_m(t) \leq t_*$, $k_*(G_m(t))=0 <1$ so that the second branch of $G_1$ gets activated.
To this end, 
\begin{align}
    G_1\!\left(G_m(t)\right)
& = \delta G_m(t) \\
& = \delta^{ (m+1) - k_*(t)} \Phi_{k_*(t)}(t). 
\end{align}
Finally, we
prove the desired claim in this case as well by noting that for $m+1 > k_*(t)$, we also have that $G_{m+1}(t)=  \delta^{ (m+1) - k_*(t)} \Phi_{k_*(t)}(t) $.

With the proof of the claim $G_1\left( G_m (t)\right)= G_{m+1}(t) $, we conclude the proof of the Proposition.
\end{proof}

\section{Containment in $\cN \in \cB^{\gamma,\delta}$} \label{Sec:Containment_class_of_channels}

The application of the previous sections requires knowing whether the channel $\cN$ is contained in $\cB^{\gamma,\delta}$. For a classical channel, the situation is very simple as the following proposition captures, which is a simple generalization of a well-known folklore result for the total variation contraction coefficient.
\begin{proposition}\label{prop:efficiency-of-classical-contract-coeff}
    Let $\gamma \geq 1$ and $\cW_{Y \vert X}$ be a classical channel. Then one may solve for $\eta_{\gamma}(\cW)$, or equivalently the minimal $\delta$ such that $\cW \in \cB^{\gamma,\delta}$ holds, in $O(\vert \cX \vert^{2}\vert \cY \vert)$ time. 
\end{proposition}
\begin{proof}
    From~\cite{asoodeh2023contractionegammadivergenceapplicationsprivacy},    $$\eta_{\gamma}(\cW) = \max_{x,x' \neq x} E_{\gamma}(\cW(\cdot \vert x) \Vert \cW(\cdot \vert x')) \  . $$
    As for classical distributions $p$ and $q$, $E_{\gamma}(p \Vert q) = \sum_{x} \max\{0,p(x) - \gamma q(x)\} + (1-\gamma)_{+}$, it is linear in $\vert \cY \vert$ to calculate $E_{\gamma}(\cW(\cdot \vert x) \Vert \cW(\cdot \vert x'))$ for an $(x,x')$ pair. There are ${\vert \cX \vert \choose 2} = \frac{1}{2}\vert \cX \vert (\vert \cX \vert-1)$ pairs to check, so combining these points, one can calculate $\eta_{\gamma}(\cW)$ in $O(\vert \cX \vert^{2}\vert \cY \vert)$. \cref{eq:ball-containment-equiv-contraction-coeff} implies the stated equivalent condition.
\end{proof}

In contrast to the above proposition, for a quantum channel $\cN$, determining $\eta_{\gamma}(\cN)$, or equivalently the minimal $\delta$ such that $\cN \in \cB^{\gamma,\delta}$ is generally quite difficult. To see this, consider that by \cref{eq:gamma-delta-ball-def}, the minimal delta is given by the solution to the bilinear optimization  
\begin{equation}
    \sup_{\rho, \sigma \in \cD(\cH)} E_{\gamma}\!\left( \cN(\rho) \Vert \cN(\sigma) \right) = \sup_{\rho,\sigma \in \cD} \sup_{0\leq M \leq I } \Tr[ M( \cN(\rho)- \gamma \cN(\sigma))] \ , 
\end{equation}
which is linear in the pair of states $(\rho,\sigma) \in \cD^{\times 2}$ for fixed POVM element $M$ and vice-versa. In general, finding the global optimum of a bilinear optimization problem is difficult, and indeed it has been shown for $\gamma =1$ (i.e. trace distance) that determining if $\eta_{\gamma=1}(\cN) = \eta_{\text{TD}}(\cN)$ is one or not is NP hard \cite{delsol2025computationalaspectstracenorm}. Given this, one would expect computing the minimal $\delta$ such that $\cN \in \cB^{\gamma,\delta}$ is generically a difficult computational problem. As such, it is useful to have computationally efficient sufficient conditions on $\delta$ given $\cN$ to guarantee $\cN \in \cB^{\gamma,\delta}$. To that end, we build on ideas from \cite{hirche2024quantum,george2025quantumdoeblincoefficientsinterpretations}, which developed tools for computational efficient bounds on $\eta_{\text{TD}}(\cN)$ via the quantum Doeblin coefficient. In particular, using the Choi operator of a superoperator $\mathcal{M}_{A\to B}$ defined as
\begin{equation}
\label{eq:choi_operator}
    \Gamma^{\cM}_{AB} \coloneqq  \sum_{i,j} |i\rangle\!\langle j|_A \otimes \cM_{A'\to B}(|i\rangle\!\langle j|_{A'}),
\end{equation}
where system $A'$ is isomorphic to system $A$, and the positive Doeblin coefficient  
\begin{equation}
\label{eq:choi_operator}
    \Gamma^{\cM}_{AB} \coloneqq  \sum_{i,j} |i\rangle\!\langle j|_A \otimes \cM_{A'\to B}(|i\rangle\!\langle j|_{A'}),
\end{equation}
 \cite[Proposition 29]{george2025quantumdoeblincoefficientsinterpretations} shows 
 \begin{align}\label{eq:positive-Doeblin-upper-bound}
     \eta_{\gamma}(\cN) \leq \sup_{r \geq 1} \sup_{\substack{\rho_{RA}\neq\sigma_{RA},\\ \rho_R=\sigma_R, \\ |R|=r}} \frac{E_{\gamma}((\id\otimes\cN)(\rho_{RA})\|(\id\otimes\cN)(\sigma_{RA}))}{E_{\gamma}(\rho_{RA}\|\sigma_{RA})} \leq 1 - \alpha_{+}(\cN) \ .
 \end{align}

\begin{proposition}
    Let $\gamma \geq 1$. If $\delta \geq 1 - \alpha_{+}(\cN)$, then $\cN \in \cB^{\gamma,\delta}$. In particular, if $\delta \geq 1 - d_{B}\lambda_{\min}(\Gamma^{\cN})$, then $\cN \in \cB^{\gamma,\delta}$.
\end{proposition}
\begin{proof}
    By the same proof as \cite[Proposition 36]{george2025quantumdoeblincoefficientsinterpretations}, $\alpha_{+}(\cN) \geq d_{B}\lambda_{\min}(\Gamma^{\cN})$. Thus, by \cref{eq:positive-Doeblin-upper-bound}, $\eta_{\gamma}(\cN) \leq 1-\alpha_{+}(\cN) \leq 1 - d_{B}\lambda_{\min}(\Gamma^{\cN})$.
\end{proof}

In the case that $\gamma \geq d_{B}$, the following result can improve upon the $1-d_{B}\lambda_{\min}(\Gamma^{\cN})$ bound.
\begin{proposition}
    If $\delta \geq 1  - \gamma \min_{\ket{v}_{A},\ket{w}_{B}}\Tr[\dyad{v} \otimes \dyad{w} \Gamma^{\cN}]$, then $\cN \in \cB^{\gamma,\delta}$. In particular, if $\cN \not \in \cB^{\gamma,0}$ and $\delta \geq 1 - \gamma \lambda_{\min}(\Gamma^{\cN})$, then $\cN \in \cB^{\gamma,\delta}$.
\end{proposition}
\begin{proof}
    If $\cN \not \in \cB^{\gamma,0}$, then by \cref{eq:contraction-optimized-by-pure-states}, \cref{eq:ball-containment-equiv-contraction-coeff}, and \cref{eq:HS_SDP},
    \begin{align}
        0 < \max_{\psi_{1} \perp \psi_{2}, 0 \leq M \leq I} \Tr[M(\cN(\dyad{\psi_{1}} - \gamma \dyad{\psi_{2}}))] \ , 
    \end{align}
    where without loss of generality $M$ is a projector. The strict inequality can only be true if the projector $M$ has rank of at least one. Continuing with the assumption $M$ is at least rank one, 
    \begin{align}
        \Tr[M(\cN(\dyad{\psi_{1}} - \gamma \dyad{\psi_{2}}))] &\leq 1 - \gamma \Tr[M\cN(\dyad{\psi_{2}})] \\
        &=1  - \gamma \Tr[\dyad{\psi_{2}}^{T} \otimes M \Gamma^{\cN}] \\
        &\leq 1  - \gamma \min_{\ket{v}_{A},\ket{w}_{B}}\Tr[\dyad{v} \otimes \dyad{w} \Gamma^{\cN}] \\
        &\leq 1 - \gamma \lambda_{\min}(\Gamma^{\cN}) \ , 
    \end{align}
    where the first inequality is $0 \leq M \leq I$ and that $\cN$ is a channel, the second inequality is that $M$ is a projector of at least rank one, and the final inequality is relaxing minimizing over product unit vectors to a unit vector.
\end{proof}

The above results in particular show non-trivial $\delta$ such that $\cN \in \cB^{\gamma,\delta}$ exist whenever $\lambda_{\min}(\Gamma^{\cN}) > 0$. The following establishes non-trivial (i.e. non-zero) lower bounds on what $\delta$ must be when $\lambda_{\min}(\Gamma^{\cN}) = 0$.

We begin again with the classical case as it both clarifies and motivates the quantum generalization.
\begin{proposition}\label{prop:cl-chan-nec-cond-on-delta}
    For a classical channel $\cW_{Y \vert X}$ and $\gamma \geq 1$, $\cW \in \cB^{\gamma,\delta}$ only if  
    \begin{align}
        \delta \geq \max_{(x,y): \cW(y\vert x) = 0} \max_{x' \neq x} \cW(y \vert x') \ .
    \end{align}
\end{proposition}
\begin{proof}
    Let $\cW(y \vert x) = 0$. Then 
    \begin{align}
        \eta_{\gamma}(\cW) &\geq E_{\gamma}(\cW(\cdot \vert x') \Vert \cW(\cdot \vert x)) \\
        &= \sup_{0 \leq M \leq I} \Tr[M(\cW(\cdot \vert x') - \gamma \cW(\cdot \vert x))] \\
        &\geq \Tr[\dyad{y}(\cW(\cdot \vert x') - \gamma \cW(\cdot \vert x))] \\
        &= \cW(y \vert x') \ , 
    \end{align}
    where the final equality uses our assumption. Maximizing over the choice of $(x,y)$ pair such that $W(y \vert x) = 0 $ and then $x' \neq x$ completes the proof.
\end{proof}

\begin{corollary}\label{cor:clas-suf-conds-for-not-being-completely-contractive}
    Let $\cW_{Y \vert X}$ such that $\cY$ is the span of the image of $\cW$. If $\lambda_{\min}(\Gamma^{\cW}) = 0$, then $\cW \not \in \cB^{\gamma,0}$ for all $\gamma \geq 1$.
\end{corollary}
\begin{proof}
Note $\lambda_{\min}(\Gamma^{\cW}) = 0$ if and only if there is an $(x,y)$ 
pair such that $\cW(y \vert x) = 0$. As $\cY$ is the span of the image of $\cW$, there exists $x' \neq x$ such that $\cW(y \vert x) > 0$. Applying the previous proposition completes the proof. 
\end{proof}
We remark restricting to the span of the image of $\cW$ is critical in the above claim. In particular, a replacer channel $\cR(\dyad{x}) = \dyad{y}$ for all $x \in \cX$ is contained in $\cB^{\gamma,0}$ for all $\gamma$. This is not a contradiction because when its output is restricted in the specified manner, the condition $\lambda_{\min}(\Gamma^{\cW}) = 0$ is not true.

We now generalize the above to the quantum scenario. This will require further notation. First, we let $\Pi_{\sigma}$ denote the projector onto the support of a state $\sigma$. Second, given a unit vector $\ket{\phi} \in A$, define $A^{\perp}_{\phi} = \text{ker}(\dyad{\phi}) = \{\ket{\psi} \in A: \langle \phi \vert \psi \rangle = 0\}$. This will be useful as 
\begin{align}\label{eq:eta-gamma-nested-optimization}
    \eta_{\gamma}(\cN) = \sup_{\ket{\phi} \in B} \; \max_{\rho \in \DD(A^{\perp}_{\phi})} E_{\gamma}(\cN(\rho) \Vert \cN(\dyad{\phi})) \ . 
\end{align}
The above follows from the right hand side being a relaxation of \cref{eq:contraction-optimized-by-pure-states}, but a restriction of \cref{eq:eta-gamma-defn} by noting $E_{\gamma}(\rho \Vert \dyad{\phi}) = 1$ for any $\ket{\phi}$ and $\rho \in \DD(A^{\perp}_{\phi})$. We may then use this to generalize \cref{prop:cl-chan-nec-cond-on-delta}.
\begin{proposition}
    For a quantum channel $\cN_{A \to B}$, $\cN \in \cB^{\gamma,\delta}$ only if for unit vector $\ket{\phi} \in B$,
    $$\delta \geq 1 - \min_{\rho \in \DD(A^{\perp}_{\phi})} \Tr[\Pi_{\cN(\phi)}\cN(\rho)] \ , $$
\end{proposition}
\begin{proof}
    Starting from \eqref{eq:eta-gamma-nested-optimization},
    \begin{align}
        \eta_{\gamma}(\cN)  &= \sup_{\ket{\tau} \in B} \max_{\rho \in \DD(A^{\perp}_{\tau})} E_{\gamma}(\cN(\rho) \Vert \cN(\dyad{\tau})) \\
        & \geq \max_{\rho \in \DD(A^{\perp}_{\phi})} E_{\gamma}(\cN(\rho) \Vert \cN(\dyad{\phi})) \\
        & = \max_{\rho \in \DD(A^{\perp}_{\phi}) , 0 \leq M \leq I} \Tr[M(\cN(\rho) - \gamma \cN(\dyad{\phi})] \\
        & \geq \max_{\rho \in \DD(A^{\perp}_{\phi})} \Tr[(I - \Pi_{\cN(\dyad{\phi})}) \cN(\rho) - \gamma \cN(\dyad{\phi})] \\
        &= 1 - \min_{\rho \in \DD(A^{\perp}_{\phi})} \Tr[\Pi_{\cN(\phi)}\cN(\rho)] \ , 
    \end{align}
    where all inequalities are making choices for optimization variables.
\end{proof}

We now prove a generalization of \cref{cor:clas-suf-conds-for-not-being-completely-contractive}, although it is not the same conditions and the proof method does not use the previous proposition.
\begin{proposition}
    Let $\cN_{A \to B}$ such that $B$ is the span of the image of $\cN$. If there is $\ket{\phi} \in A$ such that $\text{rank}(\cN(\dyad{\phi})) < d_{B}$, then $\cN \not \in \cB^{\gamma,0}$. 
\end{proposition}
\begin{proof}
    By the definition of the space $B$, there is $\rho$ such that $\cN(\rho)$ is not contained in the support of $\cN(\dyad{\phi})$. Note that $\Tr[(\rho - \gamma 
    \dyad{\phi})_{+}] > 0$ unless $\rho = \dyad{\phi}$, which would contradict that the support of $\cN(\rho)$ is not contained in the support of $\cN(\dyad{\phi})$. Thus, $E_{\gamma}(\rho \Vert \dyad{\phi}) > 0$. It follows $\eta_{\gamma}(\cN) \geq \frac{E_{\gamma}(\cN(\rho) \Vert \cN(\dyad{\phi}))}{E_{\gamma}(\rho \Vert \dyad{\phi})}$ where the lower bound is zero if and only if $E_{\gamma}(\cN(\rho) \Vert \cN(\dyad{\phi}))=0$. However,
    \begin{align}
        E_{\gamma}(\cN(\rho) \Vert \cN(\dyad{\phi})) &= \max_{0 \leq M \leq I} \Tr[\cN(\rho) - \gamma \cN(\dyad{\phi})] \\
        &\geq \Tr[(I - \Pi_{\cN(\dyad{\phi}})\cN(\rho) - \gamma \cN(\dyad{\phi})] \\
        &= \Tr[(I - \Pi_{\cN(\dyad{\phi}})\cN(\rho)] \\ 
        &= 1 - \Tr[\Pi_{\cN(\dyad{\phi})}\cN(\rho)] \\ 
        & > 0 \ , 
    \end{align}
    where the strict inequality is because $\Tr[\Pi_{\cN(\dyad{\phi})}\cN(\rho)] < 1$ as the support of $\cN(\dyad{\phi})$ does not contain the support of $\cN(\rho)$. Therefore, $\eta_{\gamma}(\cN) > 0$, so, by \cref{eq:ball-containment-equiv-contraction-coeff}, $\cN \not \in \cB^{\gamma,0}$.
\end{proof}

\section{Reverse Pinsker-type inequalities}
So far, we have shown how to build (non-linear) bounds on data processing under the hockey stick divergence (\cref{thm:non_linear_HS_div}) and its implications for non-linear SDPI. To do this, we used a non-linear relation between hockey stick divergences of different $\gamma$ parameters (\cref{prop:HS_relationship_gamma'} and its implications). When the divergence is varied for the same inputs, i.e. one considers $(\sD_{1}(\rho \Vert \sigma), \sD_{2}(\rho \Vert \sigma))$ for input pairs $(\rho,\sigma)$ and distinct divergences $\sD_{1}$ and $\sD_{2}$, then one is considering the `joint range' of these two divergences \cite{harremoes2011pairs}. From this perspective, one may identify \cref{prop:HS_relationship_gamma'} as finding a non-linear outer bound on the joint range. A similar idea to this is reverse Pinsker inequalities which upper bound a divergence $\sD$ on inputs $\rho,\sigma$ in terms of the trace distance between $\rho$ and $\sigma$, which therefore also give outer bounds on the joint range for $\sD$ and trace distance. Given this view, here we show how to use \cref{prop:HS_relationship_gamma'} with the Hirche-Tomamichel $f$-divergences to establish new reverse Pinsker-type inequalities between the $f$-divergences and the hockey stick divergences.

The main ingredients will be the following bound on the Hockey-Stick divergence that is a more fine-grained version of~\cite[Lemma 5.1]{hirche2024quantumDivergences}. 
\begin{lemma}\label{Lem:HS-Bound-HS}
    Let $1\leq\gamma_1\leq\gamma\leq\gamma_2$, then
    \begin{align}
        E_\gamma(\rho\|\sigma) \leq \frac{\gamma-\gamma_2}{\gamma_1-\gamma_2} E_{\gamma_1}(\rho\|\sigma) + \frac{\gamma_1-\gamma}{\gamma_1-\gamma_2}E_{\gamma_2}(\rho\|\sigma). 
    \end{align}
\end{lemma}
\begin{proof}
    The proof follows from the convexity of the function $\gamma\to E_\gamma$. As a result for any $1\leq\gamma_1\leq\gamma_2$ it is upper bounded by the straight line connecting the corresponding points of the function. That line can be checked to be the claimed result. 
\end{proof}
From here,~\cite[Lemma 5.1]{hirche2024quantumDivergences}, see also~\cite[Equation (23)]{zamanlooy2023strong} for the classical case, follows by setting $\gamma_1=1$ and $\gamma_2=e^{D_{\max}(\rho\|\sigma)}$. Let $f: (0,\infty) \to \RR$ be a convex and twice differentiable function 
satisfying $f(1)=0$. Then, for all quantum states $\rho$ and $\sigma$, the quantum $f$-divergence  defined in \cite[Definition~2.3]  {hirche2024quantumDivergences} is given by 
\begin{equation}\label{eq:f_divergence}
    D_f(\rho \Vert \sigma) \coloneqq \int_{1}^{\infty} f''(\gamma) E_\gamma(\rho \Vert \sigma) + \gamma^{-3} f''(\gamma^{-1}) E_\gamma(\sigma \Vert \rho) \ \mathrm{d} \gamma.
\end{equation}
These divergences reduce to the usual classical $f$-divergences for classical states. Their properties have recently been extensively discussed~\cite{hirche2024quantumDivergences,beigi2025some,liu2025layer}.
With this, we can get the following general result.
\begin{theorem}\label{thm:revPinGeneral}
    For any $\gamma_1\geq1$ and $\gamma_2\geq1$, assume that,
    \begin{align}
        E_{\gamma_1}(\rho\|\sigma) &\leq \delta_1 \\
        E_{\gamma_2}(\sigma\|\rho) &\leq \delta_2 \\
        E_1(\rho\|\sigma) &\leq\tau, 
    \end{align}
    then, 
    \begin{align}
        D_f(\rho\|\sigma) \leq f(\gamma_1) \frac{\tau-\delta_1}{\gamma_1-1} + \frac{\delta_1(f(e^a)-f(\gamma_1))}{e^a-\gamma_1} + f(\gamma_2^{-1})\frac{\gamma_2\tau-\delta_2}{\gamma_2-1} + \frac{e^b\delta_2 (f(e^{-b})-f(\gamma^{-1}))}{e^b-\gamma_2}, 
    \end{align}
    where $a = D_{\max}(\rho\|\sigma)$ and $b = D_{\max}(\sigma\|\rho)$. 
\end{theorem}
\begin{proof}
    The proof is similar to that of~\cite[Proposition 5.2]{hirche2024quantumDivergences}. Crucial is that the Hockey-Stick divergence $E_\gamma(\rho\|\sigma)$ is zero for all $\gamma\geq\exp(D_{\max}(\rho\|\sigma))$. This allows to limit the range of the integral. Then, using~\cref{Lem:HS-Bound-HS} to bound the Hockey-Stick divergence by values that we assume known, leaves us with explicitly calculating the remaining integral, which can be done directly. 
\end{proof}
In the following, we will briefly discuss some special cases and consequences of that result. 
Define, 
\begin{align}
    E_\gamma^\leftrightarrow(\rho\|\sigma) = \max\{E_\gamma(\rho\|\sigma),E_\gamma(\sigma\|\rho) \}.
\end{align}
\begin{proposition}\label{Prop:RevPinCase1}
    For any $\gamma\geq1$ and $\rho,\sigma$ such that $E_\gamma^\leftrightarrow(\rho\|\sigma)\leq\delta$, we have,
    \begin{align}
        D_f(\rho\|\sigma) \leq \frac{(\gamma+\delta)f(\gamma^{-1}) + (1-\delta)f(\gamma)}{\gamma+1} + \left[ \frac{e^b(f(e^{-b})-f(\gamma^{-1}))}{e^b-\gamma} + \frac{f(e^a)-f(\gamma)}{e^a-\gamma} \right] \delta, 
    \end{align}
    where $a = D_{\max}(\rho\|\sigma)$ and $b = D_{\max}(\sigma\|\rho)$.
\end{proposition}
\begin{proof}
From \cref{prop:HS_relationship_gamma'}, specifically the special case $\gamma=1$ previously proven in \cite[Theorem~5]{nuradha2024contraction}, we have,
    \begin{equation}
        E_{1}(\rho \Vert \sigma) \leq \frac{\gamma- 1}{\gamma+1} + \frac{2}{\gamma+1} \max\{ E_\gamma(\rho \Vert \sigma), E_\gamma(\sigma \Vert \rho)\}.
    \end{equation}
    Hence we set in \cref{thm:revPinGeneral} $\delta_1=\delta_2=\delta$ and $\tau=\frac{\gamma-1+2\delta}{\gamma+1}$ from which the results follows after some algebra. 
\end{proof}
For illustration, consider some special cases. For example in \cref{Prop:RevPinCase1}, for $\gamma=1$ and recalling $f(1)=0$, we have
\begin{align}
        D_f(\rho\|\sigma) \leq  \left[ \frac{e^bf(e^{-b})}{e^b-1} + \frac{f(e^a)}{e^a-1} \right] \delta.
\end{align}
This is exactly the reverse Pinsker inequality for $f$-divergences as proven in~\cite{hirche2024quantumDivergences}, see also~\cite{lanier2025classical} for this exact form and~\cite{binette2019note} for the classical case. Alternatively, the same result follows from \cref{thm:revPinGeneral} by setting $\delta_1=\frac{e^a-\gamma_1}{e^a-1}\tau$ and $\delta_2=\frac{e^b-\gamma_2}{e^b-1}\tau$. This also illustrates well, when that result can be more usefull than the usual reverse Pinsker inequality, namely, in particular, when $\delta_1$ and $\delta_2$ are smaller than the above thresholds. Interesting is also the case $\delta=0$ in \cref{Prop:RevPinCase1}, which gives, 
\begin{align}
        D_f(\rho\|\sigma) \leq \frac{\gamma f(\gamma^{-1}) + f(\gamma)}{\gamma+1}, 
\end{align}
which is similar to previous results in differential privacy~\cite{Christoph2024sample,nuradha2024contraction}.

\section{Mixing Times} \label{Sec:Mixing_times}

In this section, we study how the established results from previous sections can be utilized to obtain mixing times by having the hockey-stick divergence as the divergence of interest. For trace distance ($\gamma=1$), mixing times of several noisy channels have been studied in~\cite{levin2017markov,Raginsky2016,george2025quantumdoeblincoefficientsinterpretations}.


Define mixing time for $E_\gamma$ as follows: for $\gamma \geq 1$, $\beta\in[0,1]$, and channel $\cN$ with the fixed point $\sigma^*$
\begin{equation}
    t_\gamma^\cN(\beta) \coloneqq \min\left\{ n \in \mathbb{N}: \sup_{\rho}  E_{\gamma}\!\left( \cN^{(n)}(\rho) \Vert \sigma^*\right) \leq \beta  \right\}, 
\end{equation}
where $\cN^{(n)} \coloneqq \underbrace{\cN \circ \cdots \circ \cN}_{n}$.

\begin{proposition} \label{prop:mixing_times1}
    Let $\gamma \geq \gamma' > 1$ and $\beta \in [0,1]$. For a channel $\cN \in \cB^{\gamma, \delta=0}$ has the following upper bound on the mixing time: 
    \begin{equation}
        \label{Eq:tNbeta}
        t_{\gamma'}^\cN (\beta) 
        \leq  \left \lceil \frac{ \ln\!\left( \frac{\gamma' +1}{2 \beta + \gamma' -1}\right) }{\ln\!\left( \frac{\gamma +1}{\gamma -1}\right)} \right \rceil.
    \end{equation}
\end{proposition}

\begin{proof}
 First, note that since $\sigma^*$ is a fixed point of the channel $\cN$, we have $\cN(\sigma^*)= \sigma^*$.
By \cref{prop:F_gamma_n_channels}, for all $\rho$,
\begin{align*}
    E_{\gamma'}\!\left( \cN^{(n)}(\rho) \Vert \sigma^*\right) = E_{\gamma'}\!\left( \cN^{(n)}(\rho) \Vert \cN^{(n)} (\sigma^*)\right)  &\leq \frac{1}{2} \left(-(\gamma'-1) + \left(\frac{\gamma -1}{ \gamma +1} \right)^n (\gamma'+1)  \right)_+ \ .
\end{align*}
As $\gamma' > 1$, for all $\beta \in [0,1]$, there exists $n_\beta \in \mbb{N}$ such that $-(\gamma'-1)  + \left(\frac{\gamma -1}{ \gamma +1} \right)^n (\gamma'+1) \leq 2\beta$ for all $n \geq n_{\beta}$. Thus, solving $n$ as a function of $\beta$ in this inequality, we find the first upper bound.
\end{proof}

Next, we study mixing times of channels that belongs to the set $\cB^{\gamma,\delta}$ with $\delta >0$.

\begin{proposition} \label{prop:mixing_linear_gam_del}
    Let $\gamma \geq \gamma' > 1$ and $\beta \in (0,1)$. For  $\cN \in \cB^{\gamma, \delta}$ and $\beta \in (0,1)$, we have 
\begin{equation}\label{eq:tNlinear}
     t_{\gamma'}^\cN(\beta) \leq \left \lceil \frac{ \ln\!\left( \frac{1}{\beta}\right) }{\ln\!\left( \frac{\gamma +1}{\gamma -\gamma' +(1+\gamma') \delta}\right)} \right \rceil.
\end{equation}
\end{proposition}
\begin{proof}
       For $\cN \in \cB^{\gamma, \delta}$, by applying~\cref{prop:contraction_coeff_upper_bound} $n$ times for each application of the channel $\cN$ in the sequential composition, we get
    \begin{equation}
        \sup_{\rho} E_{\gamma'}\!\left( \cN^{(n)}(\rho) \Vert \cN^{(n)} (\sigma^*)\right) \leq \left(\frac{(\gamma- \gamma') + \delta (\gamma'+1) }{\gamma+1} \right)^n .
    \end{equation}
    Then, upper-bounding the above by $\beta$ and solving for $n$ provides an upper bound on the mixing time.
\end{proof}

\begin{remark}[Strength of Non-Linear SDPI for Obtaining Mixing Times] \label{rem:compare_mixing}
    Note that  for $\delta=0$, \eqref{Eq:tNbeta} results in a finite value even when $\beta = 0$ whereas \eqref{eq:tNlinear} is unbounded when $\beta = 0$. This is an example of the fact that for hockey stick divergences with $\gamma > 1$, nonlinear SDPI can determine mixing times such that the divergence becomes zero, but linear contraction coefficients cannot.

    Furthermore, in Figure~\ref{fig:mixing_linear_Non}, we illustrate how non-linear SDPI can lead to tighter mixing times compared to the bounds derived via linear SDPI, specifically when $\beta$ is small (i.e.; $\beta \ll 1$).
\end{remark}
\begin{figure}
    \centering
    \includegraphics[width=0.8\linewidth]{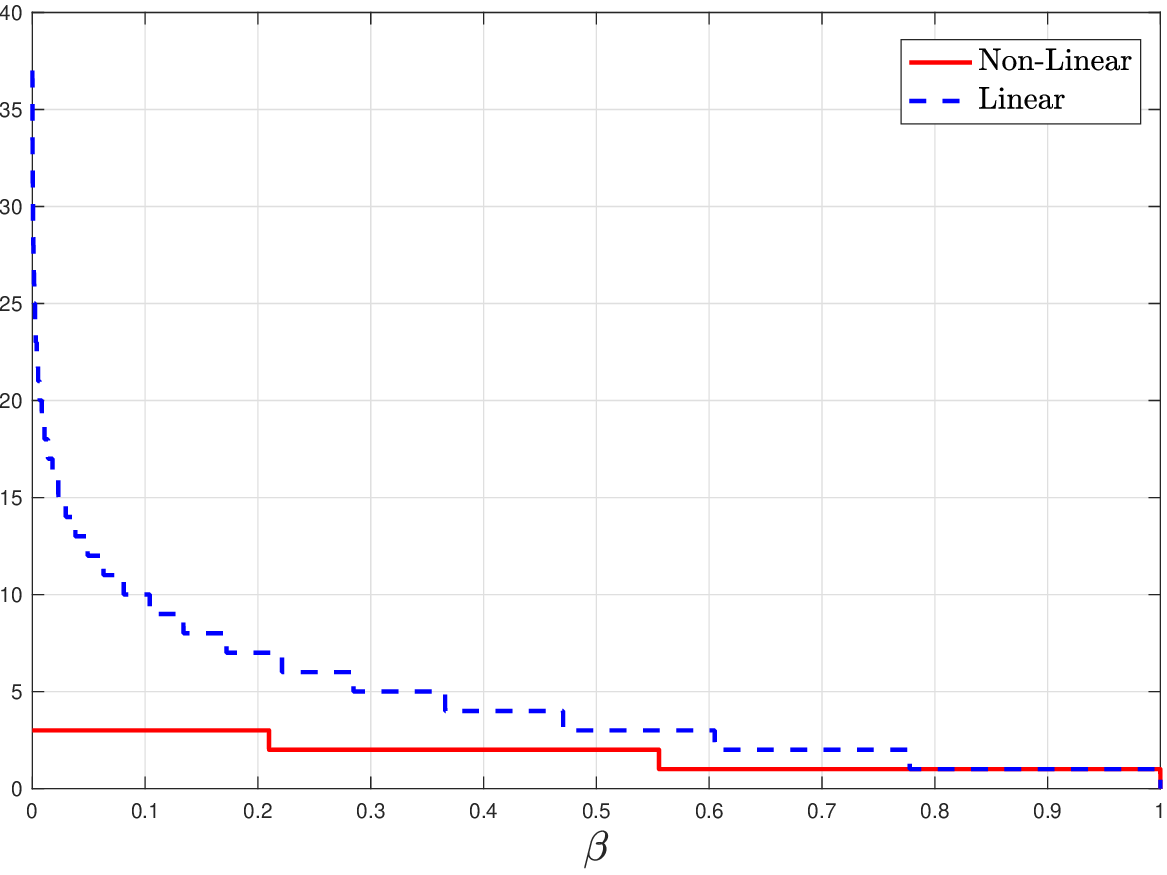}
    \caption{Comparison of Upper Bounds on Mixing Times with Linear and Non-Linear SDPI: We consider $\cN \in \cB^{\gamma, \delta}$ with $\gamma=8$ and $\delta=0$. With that, we plot the upper bounds on $t_{\gamma'}^{\cN}(\beta)$ given in~\cref{prop:mixing_times1} (marked as Non-linear) and~\cref{prop:mixing_linear_gam_del} (marked as Linear) for $\gamma'=3$.}
    \label{fig:mixing_linear_Non}
\end{figure}

So far, we have only addressed the case $\cN \in \cB^{\gamma,\delta}$  with $\delta \neq 0$ using a linear SDPI. Here we use our non-linear SDPI result (Proposition \ref{prop:F_gamma_n_delta}) to get a better understanding of the mixing times together with any transitions that may arise. 

\begin{proposition}
\label{prop:mixing_times_delta_not_0}
Let $1 < \gamma' <\gamma$, $\delta \in (0,1)$ and $\cN \in \cB^{\gamma, \delta}$. For $\beta \in (0,1)$, we have that 
\begin{equation}
        t_{\gamma'}^{\cN}(\beta)\leq
    k_* +  \left(\left\lceil\frac{\ln(\beta/T_*)}{\ln\delta}\right\rceil \right)_+,
\end{equation}
where $k_*\coloneqq k_*(1)$ and
  $T_*\coloneqq \Phi_{k_*}(1)$ by choosing $t=1$ in~\eqref{eq:kstar} and~\eqref{eq:phik}. 
\end{proposition}
\begin{proof}
    First notice that for all $\rho \in \cD$ (since $E_\gamma(\rho \Vert \sigma^*) \leq 1$)
    \begin{align}
      E_{\gamma'}\!\left( \cN^{(n)}(\rho) \Vert \cN^{(n)} (\sigma^*)\right)  \leq F_{\gamma'}^{\cN^{(n)}}(1) \leq G_n(1),
    \end{align}
 where the first inequality follows from~\eqref{eq:f_gamma_def}and the second from~\cref{prop:F_gamma_n_delta} with 
  \begin{equation}
 G_n(1) =
    \begin{cases}
      \Phi_n(1), & 1\leq n\leq  k_*(1),\\
      \delta^{\,n-k_*(1)} \Phi_{k_*(1)}(1), & n \geq k_*(1).
    \end{cases}
    \label{eq:Gn-piecewise_1}
  \end{equation}
Then by choosing $n$ such that $G_n(1) \leq \beta$ provides the desired upper bound. To this end, we consider two cases. 

\textbf{Case 1:} For the setting $\beta \ge T_*$ since $G_{k_*}(1)=T_*$ and $G_n(1)$ is non-increasing in $n$ after $k_*$ (geometric decay on the branch with $\delta$ dependence),
we already have $G_{k_*}(1)\leq \beta$. This leads to
\begin{equation}
    t_{\gamma'}^{\cN}(\beta) \leq k_*
\end{equation}
This matches the bound stated because $\ln(\beta/T_*)\le 0$ and $\ln\delta<0$ imply
$\big\lceil \ln(\beta/T_*)/\ln\delta \big\rceil \le 0$, yielding desired upper bound $k_*$.

\textbf{Case 2:} For $0<\beta < T_*$, then $n>k_*$ with $G_n(1)=\delta^{\,n-k_*}\,T_*$. The condition $G_n(1)\le \beta$ is equivalent to $ \delta^{\,n-k_*(1)} \Phi_{k_*(1)}(1) \leq \beta$.
This leads to 
\begin{equation}
    n \geq k_* + \frac{\ln(\beta/T_*)}{\ln\delta}.
\end{equation}
Thus the smallest integer $n$ satisfying $G_n(1)\le \beta$ is
\begin{equation}
    n = k_* + \left\lceil \frac{\ln(\beta/T_*)}{\ln\delta} \right\rceil,
\end{equation}
and consequently
\begin{equation}
    t_{\gamma'}^{\cN}(\beta) \leq k_* + \left\lceil \frac{\ln(\beta/T_*)}{\ln\delta} \right\rceil .
\end{equation}
Combining the two cases concludes the stated upper bound on the mixing time.

\begin{remark}[Non-linear Strong Data-Processing towards Mixing]
    By~\cref{prop:mixing_times_delta_not_0}, we see that the mixing time depends on two different parameters $k_*$ and $T_*$, where $k_*\coloneqq k_*(1)$ and
  $T_*\coloneqq \Phi_{k_*}(1)$ by choosing $t=1$ in~\eqref{eq:kstar} and~\eqref{eq:phik}. The upper bounds on mixing times have two different scalings depending on whether $\beta \geq T_*$ or $\beta <T_*$. In fact, for $\beta >T_*$, the upper bound on the mixing time is independent of $\beta$, and we see a transition at $\beta=T_*$.
\end{remark}



\end{proof}

As highlighted in Remark \ref{rem:compare_mixing}, a particularly appealing property of nonlinear SDPI for hockey stick divergences of parameter $\gamma' > 1$ is that they can derive finite mixing times for the case $\beta = 0$, i.e. a finite point at which the initial input state and the fixed point of the channel are indistinguishable with respect to the hockey-stick divergence. It was shown in~\cite[Theorem 2]{zamanlooy2024mathrm} that when a classical channel $\cW$ satisfies the condition $\sup_{p,q} D_{\max}(\cW(p) \Vert \cW(q)) = a < +\infty$, there exist simple bounds on the mixing time $t^{\cW}_{\gamma}(0)$. We end this section by generalizing this result and highlighting for what sort of channel this is relevant.

To generalize \cite[Theorem 2]{zamanlooy2024mathrm}, we require defining 
\begin{align}
    a  \equiv a(\cN) & \coloneq \sup_{\rho,\sigma} D_{\max}(\cN(\rho)\|\cN(\sigma)) \\
    \zeta(a) &\coloneq \frac{a-1}{a+1}. \label{eq:define_zeta}
\end{align}
One may observe that $\cN$ is always in $\cB^{e^{a}, 0}$. One may also observe that $a(\cN)$ is finite if and only if the support of all output states is the same, as follows from $D_{\max}(\rho \Vert \sigma) < +\infty$ only if the support of $\rho$ is contained in the support of $\sigma$. We remark that this is a strong condition on the channel $\cN_{A \to B}$ as it means all states are full rank on the span of the image of $\cN$ after a single iteration. In particular, this means any channel $\cN_{A \to B}$ such that $a(\cN) < +\infty$ is ``effectively'' a special case of a primitive channel \cite{wolf2012quantum}. To see this, given a channel $\cN$ such that $a(\cN) < +\infty$, we define $B'$ as the support of $\cN(\dyad{0})$ and $\Pi_{B'}$ as the projector onto $B'$. Then considering the Kraus decomposition of $\cN(\cdot) = \sum_{i} K_{i} \cdot K_{i}^{\ast}$, we may define the channel $\cM_{B' \to B'}(\cdot) = \sum_{i} K_{i}\Pi_{B'} \cdot \Pi_{B'}K_{i}^{\ast}$. By definition of $\cN$ and $B'$, for all input states, $\cN^{(n)} = \cM^{(n-1)} \circ \cN$. Moreover, as $a(\cN) < +\infty$, $a(\cM) < +\infty$, and thus maps all input states to full rank states in a single step. As a primitive channel is a channel that eventually maps all states to full rank states and equivalently converges to a unique full rank state \cite[Theorem 6.7]{wolf2012quantum}, $\cM$ is a specific case of a primitive channel that maps all states to a full rank state in a single step. Thus, as $\cN^{(n)}$ acts as $\cM$ for every iteration but the first, we can view $\cN$ as ``effectively" a stronger notion of a primitive channel, and thus the demand $a(\cN) < +\infty$ as selecting for a stronger condition than being primitive. We also remark the above shows if $\sigma^{\ast}$ is the fixed point of $\cN$, then $\sigma^{\ast}$ is full rank on $B'$.

With the above considerations specified, we study the mixing times of channels where $a(\cN)$ is finite. First, we observe what one can deduce directly from our earlier results.
\begin{corollary} \label{Cor:mixingtime2}
    Let $\cN_{A \to B}$ such that $a(\cN) < +\infty$ and $\sigma^{\ast}$ be its fixed point. For any $\gamma'$ such that $(\gamma'+1)\lambda_{\min}(\sigma^{\ast})\leq 1$, 
        \begin{align}
        t_{\gamma'}^\cN(0) \leq  \left \lceil \frac{ \ln\!\left( (\gamma' -1)\lambda_{\min}(\sigma^{\ast})\right) }{\ln\!\left( \zeta(e^a)\right)} \right \rceil,
    \end{align}
    where $\zeta(\cdot)$ is defined in~\eqref{eq:define_zeta}.
\end{corollary}
\begin{proof}
    This follows directly from Equation~\eqref{Eq:tNbeta}, by setting $\gamma=\exp(a)$, $\beta=0$ and using the assumed bound on $\gamma'$.  
\end{proof}
This can be contrasted with~\cite[Theorem~2]{zamanlooy2024mathrm}, for which we can also give a quantum generalization as follows. 
\begin{corollary} \label{Cor:mixingtime3}
Let $\cN_{A \to B}$ such that $a(\cN) < +\infty$ and $\sigma^{\ast}$ be its fixed point. For any $\gamma'$ such that $(\gamma'-1)\lambda_{\min}(\sigma)\leq 1$, 
        \begin{align}
        t_{\gamma'}^\cN(0) \leq  \left \lceil \frac{ \ln\!\left( (\gamma' -1)\lambda_{\min}(\sigma^{\ast})\right) }{\ln\!\left( \zeta(e^a)\right)} \right \rceil.
    \end{align} 
\end{corollary}
\begin{proof}
    We follow the proof of~\cite[Theorem~2]{zamanlooy2024mathrm}. Starting from Equation~\eqref{eq:contraction_hockey_delta_0}, we have
    \begin{align} 
        \eta_{1}(\cN) \leq \left(\frac{e^a -1}{e^a +1} \right) = \zeta(e^a) .
    \end{align}
    Hence,
    \begin{align}
        E_{1}\!\left( \cN^{(n)}(\rho) \Vert \sigma^*\right) \leq  \eta_{1}(\cN)^n \leq  \zeta(e^a)^n.        
    \end{align}
    Now choosing, 
    \begin{align}
        n=\frac{\log\left[ (\gamma'-1)\lambda_{\min}(\sigma^{\ast})\right]}{\log\left[\zeta(e^a) \right]},
    \end{align}
    leads to
        \begin{align}
        E_{1}\!\left( \cN^{(n)}(\rho) \Vert \sigma^*\right) \leq (\gamma'-1)\lambda_{\min}(\sigma^{\ast}),         
    \end{align}
    which then implies, by Proposition~\ref{Prop:gamma-gamma-ev},
    \begin{align}
      E_{\gamma'}\!\left( \cN^{(n)}(\rho) \Vert \sigma^*\right) =0, 
    \end{align}
    from which the results follows.
\end{proof}
That \Cref{Cor:mixingtime3} is applicable to a larger range of $\gamma'$ and hence slightly more general, possibly due to its more specialized proof.

\section{Quantum Privacy}
\label{Sec:privacy}

With the development of quantum technologies and the generation of quantum data, ensuring the privacy of quantum systems is an important direction to explore. To this end,
statistical privacy frameworks for quantum data have been developed recently~\cite{QDP_computation17, aaronson2019gentle,hirche2022quantum, nuradha2023quantum}, which are generalizations of classical statistical privacy frameworks, including differential privacy~\cite{DMNS06,DR14} and pufferfish privacy~\cite{KM14,nuradha2022pufferfishJ}. One such framework is quantum local differential privacy (QLDP), which ensures that two distinct quantum 
states passed through a quantum channel are hard to distinguish by a measurement~\cite{hirche2022quantum}.  QLDP is also a special case of quantum pufferfish privacy~\cite{nuradha2023quantum}, and it has been studied in~\cite{angrisani_localModel25,guan2024optimal,nuradha2024contraction,Christoph2024sample} with respect to its performance tradeoffs on various statistical tasks. Also, the classical setting of local differential privacy has been studied in~\cite{zamanlooy2024mathrm} by utilizing non-linear SDPI for hockey-stick divergences.
\begin{definition}[Quantum Local Differential Privacy]\label{def:QLDP}
    Fix $\varepsilon \geq 0$ and $\delta \in [0,1]$. 
    Let $\cA$ be a quantum algorithm (viz., a quantum channel). The algorithm~$\cA$ is $ (\varepsilon, \delta)$-local differentially private if 
\begin{equation} \Tr\!\left[M \cA(\rho)\right] \leq e^\varepsilon \Tr\!\left[M \cA(\sigma)\right] + \delta,\qquad \forall  \rho, \sigma \in \cD(\cH), \quad \forall M: 0\leq M \leq I.
\label{eq:QLDP-def}
\end{equation}
We say that $\cA$ satisfies $\varepsilon$-QLDP if it satisfies $(\varepsilon,0)$-QLDP.
\end{definition}
Note that private channels that satisfy quantum local differential privacy (QLDP) as defined in~\cref{def:QLDP} fits into the noisy channel criterion in~\eqref{eq:intro_class_of_channels} by choosing $\gamma= e^\varepsilon$ therein~\cite[Eq. (V.1)]{hirche2022quantum}.

\subsection{Stronger Privacy Guarantees from Composition}
In this section, we study how we can improve the privacy guarantees of quantum channels imposed by QLDP by utilizing the SDPI derived in this work. 
We ask the following question: How can the composition results be improved in contrast to DPI and contraction coefficients?

Let $\cA$ be $(\varepsilon, 0)$- QLDP mechanism with $\cA : \cL(\cH_A) \to \cL(\cH_B)$ with $d_A=d_B$. Then, by data-processing property, we have that $\cA \circ \cA$ also 
satisfies $(\varepsilon,0)$- QLDP. Furthermore, by using contraction coefficients of channels, we can obtain strengthened privacy parameters such that they satisfy $(\varepsilon',   \left(\eta_{e^{\varepsilon'}}(\cA) \right)^2 )$ with $\varepsilon' \leq \varepsilon$, due to 
\begin{align}
    \sup_{\rho, \sigma} E_{e^{\varepsilon'}}\!\left( \cA \circ \cA (\rho) \Vert \cA \circ  \cA (\sigma) \right) & \leq \eta_{e^{\varepsilon'}}(\cA) \  \sup_{\rho, \sigma} E_{e^{\varepsilon'}}\!\left( \cA(\rho) \Vert \cA(\sigma) \right) \\ 
    & \leq \left(\eta_{e^{\varepsilon'}}(\cA) \right)^2 \sup_{\rho,\sigma} E_{e^\varepsilon}(\rho \Vert \sigma) \\
    & =\left(\eta_{e^{\varepsilon'}}(\cA) \right)^2 \\
    & \leq \left(\frac{e^\varepsilon - e^{\varepsilon'}}{e^{\varepsilon} +1} \right)^2,
\end{align}
where the penultimate inequality follows by noticing that the supremum is achieved by a pair of orthogonal states and evaluates to one; 
and the last inequality by~\eqref{eq:contraction_hockey_delta_0} with $\gamma=e^\varepsilon$ and $\gamma'=e^{\varepsilon'}$.

We show next that by employing the results derived in this work by establishing both linear and non-linear contraction coefficients provides improved sequential composition guarantees even for the setting with $\delta \neq 0$. Also, assume that all channels considered in the rest of this section have the same input and output dimensions (i.e.; $\cA : \cL(\cH_A) \to \cL(\cH_B)$ with $d_A=d_B$). 

\begin{proposition} \label{prop:strong_privacy_eps_delta}
Let $\varepsilon \geq 0$. 
    Let $\cA_i$ for $i=\{1, \ldots, n\}$ satisfies $(\varepsilon, 0)$-QLDP. Then, the sequential composition of $\cA_1 \circ \cdots \circ \cA_n$ satisfies $(\varepsilon', \delta')$-QLDP, where $0 <\varepsilon' \leq \varepsilon$
     \begin{equation}
        \delta' \coloneqq  \min \left\{ \frac{1}{2} \left((\zeta(\varepsilon))^n (e^{\varepsilon'}+1) +1 -e^{\varepsilon'} \right)_+, \left( \frac{e^\varepsilon- e^{\varepsilon'}  }{e^\varepsilon +1} \right)^n \right\}
    \end{equation}
and 
\begin{equation}
    \zeta(\varepsilon) \coloneqq \frac{e^\varepsilon -1}{e^\varepsilon +1}.
\end{equation}
Furthermore, let $\cA_i$ be $(\varepsilon_i,0)$-QLDP, where $\varepsilon_i \geq 0$. Then, we have that $\cA_1 \circ \cdots \circ \cA_n$ satisfies $(\varepsilon^*, 0)$-QLDP, where 
\begin{equation}
    \varepsilon^* \coloneqq \frac{1+\prod_{i=1}^n  \zeta(\varepsilon_i)}{1 - \prod_{i=1}^n  \zeta(\varepsilon_i)}.
\end{equation}

\end{proposition}

\begin{proof}
The proof follows by using~\cref{prop:F_gamma_n_channels} for $\delta=0$ and~\eqref{eq:HS_Con_delta_0}, so as to obtain that for $ \gamma' \geq 1$ together with $E_{\gamma'}(\rho \Vert \sigma) \leq 1$
\begin{equation}
   \sup_{\rho, \sigma} E_{\gamma'}\!\left( \cA^{(n)} (\rho) \Vert \cA^{(n)} \sigma\right) \leq \min \left\{F_{\gamma'}^{\cN^{(n)}}(1), \left( \left(\frac{\gamma- \gamma'  }{\gamma+1} \right)_+\right)^n  \right \}.
\end{equation}
We conclude the proof of the first statement by choosing $\gamma' = e^{\varepsilon'}$ and $\gamma= e^{\varepsilon}$.
For the second setting, proof follows similarly to the first setting, by using~\cref{prop:F_gamma_hetero_comp} instead of ~\cref{prop:F_gamma_n_channels}.
\end{proof}

\begin{corollary}\label{Cor:QLDP_2}
    Let $\varepsilon \geq 0$ and $\delta \in [0,1]$.
    Let $\cA_i$ for $i=\{1, \ldots, n\}$ satisfies $(\varepsilon, \delta)$-QLDP. Then, the sequential composition of $\cA_1 \circ \cdots \circ \cA_n$ satisfies $(\varepsilon', \delta')$-QLDP, where $\varepsilon' \leq \varepsilon$ and
     \begin{equation}
        \delta' \coloneqq  \left( \frac{(e^\varepsilon- e^{\varepsilon'} ) + \delta (e^{\varepsilon'}+1) }{e^{\varepsilon}+1} \right)^n.
    \end{equation}
\end{corollary}
\begin{proof}
    Proof follows by sequentially applying~\cref{prop:contraction_coeff_upper_bound} repeatedly $n$ times.
\end{proof}

We next show that by composition of $(\varepsilon, \delta)$-QLDP mechanisms, it is possible to generate mechanisms satisfying $(\varepsilon',0)$-QLDP, generalizing~\cite[Lemma~3]{zamanlooy2024mathrm} for the quantum setting. 
\begin{proposition} \label{prop:delta_zero_mechanisms}
    Let $\cA$ be $(\varepsilon,\delta)$-QLDP with $\varepsilon >0$ and $\delta \in (0,1)$, and $\sigma^*$ be a fixed point of $\cA$ with $\lambda_{\min}(\sigma^*) >0$. Then, applying $\cA$ $n$-times sequentially produces a channel that satisfies $(\varepsilon',0)$-QLDP with $0 <\varepsilon' <\varepsilon $ when 
    \begin{equation}
        n \geq \max\left\{   \frac{\ln\!\left( (e^{\varepsilon'} -1)  \lambda_{\min}(\sigma^*) /2 \right)}{\ln (f_{\varepsilon,\delta})}, \frac{\ln\!\left(   \lambda_{\min}(\sigma^*)/2 \right)}{\ln (f_{\varepsilon,\delta})}\right\},
    \end{equation}
where 
\begin{equation}
    f_{\varepsilon,\delta} \coloneqq \left(\frac{e^{\varepsilon} -1 + 2 \delta }{e^{\varepsilon} +1} \right).
\end{equation}
\end{proposition}
\begin{proof}
    First, by choosing $n$ large enough such that $n \geq \frac{\ln \left(   \lambda_{\min}(\sigma^*)/2 \right)}{\ln (f_{\varepsilon,\delta})}$, we have that 
    \begin{equation} \label{eq:n_first_bound_eps_delt}
        (f_{\varepsilon,\delta})^n \leq \frac{\lambda_{\min}(\sigma^*)}{2}.
    \end{equation}
    Note that $f_{\varepsilon,\delta} <1$ for $\varepsilon >0$ and $\delta \in (0,1)$ and there exists $n$ such that the above inequality holds.
    
    Let $\omega$ be an arbitrary quantum state where  $|v\rangle$  is the eigenvector achieving 
    \begin{equation}
        \lambda_{\min}\!\left( \cA^{(n)}(\omega)\right) = \langle v| \cA^{(n)}(\omega) | v \rangle.
    \end{equation}
    Then, we also have that 
    $ \lambda_{\min}(\sigma^*) = \inf_{\tau \in \cD}  \Tr[ \tau \sigma^*]\leq \langle v| \sigma^* | v \rangle$.
    With that, we consider 
    \begin{align}
        \lambda_{\min}(\sigma^*) -  \lambda_{\min}\!\left( \cA^{(n)}(\omega)\right)  & \leq   \langle v| \sigma^* | v \rangle- \lambda_{\min}\!\left( \cA^{(n)}(\omega)\right)  \\ 
        &=  \Tr\! \left[ |v \rangle \!\langle v| ( \sigma^* - \cA^{(n)}(\omega))\right] \\
        & \leq T\!\left(\sigma^*,\cA^{(n)}(\omega) \right) \\
        &=T\!\left( \cA^{(n)}(\sigma^*),\cA^{(n)}(\omega) \right) \\
        & \leq \left(\frac{e^{\varepsilon} -1 + 2 \delta }{e^{\varepsilon} +1} \right)^n T(\sigma^*,\omega) \\
       & \leq \left(\frac{e^{\varepsilon} -1 + 2 \delta }{e^{\varepsilon} +1} \right)^n,
    \end{align}
where the second inequality follows by the SDP formulation of trace distance by noticing that $0 \leq |v \rangle \!\langle v| \leq I$, penultimate inequality by applying SDPI inequality for $\gamma'=1$ and $\gamma=e^\varepsilon$ in~\cref{prop:contraction_coeff_upper_bound} for $n$ consecutive times, and the last inequality by bounding the trace distance between two states by one.
Together with~\eqref{eq:n_first_bound_eps_delt} we have that for all $\omega \in \cD$ (recall that $\omega$ was chosen arbitrarily)
\begin{equation}
    \frac{ \lambda_{\min}(\sigma^*)}{2}  \leq \lambda_{\min}\!\left( \cA^{(n)}(\omega)\right).
\end{equation}

Let $\rho, \omega$ be arbitrary states.  then Since $\cA$ is $(\varepsilon, \delta)$-QLDP, we have 
\begin{equation}
    T\!\left( \cA^{(n)}(\rho) , \cA^{(n)}(\omega)\right) \leq  \left(\frac{e^{\varepsilon} -1 + 2 \delta }{e^{\varepsilon} +1} \right)^n,
\end{equation}
by applying SDPI inequality for $\gamma'=1$ and $\gamma=e^\varepsilon$ in~\cref{prop:contraction_coeff_upper_bound} for $n$ consecutive times, and bounding the trace distance between inputs by one. 
Now by choosing $n \geq \frac{\ln\!\left( (e^{\varepsilon'} -1)  \lambda_{\min}(\sigma^*) /2 \right)}{\ln (f_{\varepsilon,\delta})}$ in addition to $n \geq \frac{\ln \left(   \lambda_{\min}(\sigma^*)/2 \right)}{\ln (f_{\varepsilon,\delta})}$, we get 
\begin{align}
     T\!\left( \cA^{(n)}(\rho) , \cA^{(n)}(\omega)\right) & \leq  (e^{\varepsilon{'}} -1) \  \frac{\lambda_{\min}(\sigma^*)}{2}\\
     & \leq  (e^{\varepsilon{'}} -1) \ \lambda_{\min}\!\left( \cA^{(n)}(\omega)\right). 
\end{align}
With the above inequality, by applying~\cref{Prop:gamma-gamma-ev} with the choice $\gamma'=1$ and $\gamma=e^{\varepsilon'}$, we have that $ E_{e^{\varepsilon'}}\left( \cA^{(n)}(\rho) , \cA^{(n)}(\omega)\right) =0$. Since $\rho$ and $\omega$ was chosen arbitarily, we have that 
\begin{equation}
    \sup_{\rho,\omega \in \cD} E_{e^{\varepsilon'}}\left( \cA^{(n)}(\rho) , \cA^{(n)}(\omega)\right) =0,
\end{equation}
implying that $\cA^{(n)}$ is $(\varepsilon',0)$-QLDP, whenever $n$ is chosen as given in the proposition statement.
\end{proof}

\begin{remark}[Generating stronger privacy frameworks by composition] 
   \cref{prop:strong_privacy_eps_delta} provides a provable privacy guarantees on the composition of $n$  less private mechanisms (with $(\varepsilon, \delta)$-QLDP) to obtain a stronger privacy mechanism having $\varepsilon' \leq \varepsilon$ while $0 \leq \delta' \leq \delta$. 
   
    As a special case, by utilizing~\cref{prop:strong_privacy_eps_delta}, one can obtain a private mechanism with  $\epsilon' \leq \epsilon$ and  $\delta' =0$  
    with repeated sequential application ($n$ times) of a private mechanism satisfying $(\varepsilon,0)$-QLDP by choosing
   \begin{equation}
       n \geq \left \lceil \frac{ \ln\!\left( \frac{e^{\varepsilon'} +1}{e^{\varepsilon'} -1}\right) }{\ln\!\left( \frac{1}{\zeta}\right)}  \right\rceil.
   \end{equation}
   Also, from~\cref{prop:delta_zero_mechanisms}, one can even generate mechanisms having $\delta'=0$ starting with $(\varepsilon, \delta)$-QLDP mechanisms whenever the channel has a full-rank fixed-point. To this end, depolarizing channels satisfy the phenomenon since the maximally mixed state is a fixed point. 
\end{remark}

\subsection{Bounds on $f$-Divergences under Privacy Constraints}

Finally, in this section, we will use some of the previous results to give bounds on $f$-divergences under differential privacy constraints, which are valid for all $\delta \in [0,1]$. 
The main result is the following. 
\begin{proposition} \label{prp:f_div_Contraction}
    Let $\cN$ be an $(\epsilon,\delta)$-QLDP channel, then
    \begin{align}
        D_f(\cN(\rho)\|\cN(\sigma)) \leq &\frac{f(e^\epsilon)+e^\epsilon f(e^{-\epsilon})}{e^\epsilon-1}\frac{e^\epsilon-1+2\delta}{e^\epsilon+1}\tau - \frac{f(e^\epsilon)+f(e^{-\epsilon})}{e^\epsilon-1}\delta \nonumber\\
        &+ \lambda \left( f\left(e^\epsilon+\frac{\delta}{\lambda}\right) - f\left(e^\epsilon\right) + \left(e^\epsilon+\frac{\delta}{\lambda}\right)f\left(\left(e^\epsilon+\frac{\delta}{\lambda}\right)^{-1}\right)-\left(e^\epsilon+\frac{\delta}{\lambda}\right)f\left(e^{-\epsilon}\right) \right),
    \end{align}
    where $\tau=E_1(\rho\|\sigma)$ is the trace distance and $\lambda=\inf_\sigma\lambda_{\min}(\cN(\sigma))$. 
\end{proposition}
\begin{proof}
    The main ingredient is~\cref{thm:revPinGeneral}, which we apply with $\gamma_1=\gamma_2=e^\epsilon$ and $\delta_1=\delta_2=\delta$. The resulting  $E_1(\cN(\rho)\|\cN(\sigma))$ we bound further with~\cref{Eq:E1-contraction}, which leads to the first term in the claim. The second term is immediate. For the unwieldy third term, we need to also bound the max-relative entropy. For that we use~\cref{Cor:Dmax-by-smooth}, which gives, 
    \begin{align}
    e^{D_{\max}(\rho\|\sigma)} \leq e^{D^\delta_{\max}(\rho\|\sigma)} + \frac{\delta}{\lambda_{\min}(\sigma)} \leq e^{\epsilon} + \frac{\delta}{\lambda_{\min}(\sigma)}.
\end{align}
To apply this, note that the bound in~\cref{thm:revPinGeneral} is monotonically non-decreasing in $a$ and $b$. Since we don't make any further assumptions on $\sigma$, we choose $\lambda=\inf_\sigma\lambda_{\min}(\cN(\sigma))$, such that the previous relation holds for all possible output states of the channel. Bringing all the above steps together completes the proof.
\end{proof}
This is a very general result. To get a better idea, we will compare some special cases. First note that for $\delta=0$, we recover
    \begin{align}
        D_f(\cN(\rho)\|\cN(\sigma)) \leq &\frac{f(e^\epsilon)+e^\epsilon f(e^{-\epsilon})}{e^\epsilon+1}\tau,
    \end{align}
which was previously shown in~\cite[Proposition 5]{nuradha2024contraction}. For further comparison it will be helpful to specialize to the relative entropy.
\begin{corollary}\label{Cor:RE-LDP-bound}
       Let $\cN$ be an $(\epsilon,\delta)$-QLDP channel, then
    \begin{align}
        D(\cN(\rho)\|\cN(\sigma)) \leq &\epsilon\frac{e^\epsilon-1+2\delta}{e^\epsilon+1}\tau - \epsilon\frac{e^\epsilon-e^{-\epsilon}}{e^\epsilon-1}\delta \nonumber\\
        &+ \lambda \left( \left(e^\epsilon+\frac{\delta}{\lambda}-1\right)\log\left(e^\epsilon+\frac{\delta}{\lambda}\right) + \left(1-e^\epsilon+\frac{\delta}{\lambda}e^{-\epsilon}\right)\epsilon \right),\label{Eq:RE-LDP-bound}
    \end{align}
    where $\tau=E_1(\rho\|\sigma)$ is the trace distance and $\lambda=\inf_\sigma\lambda_{\min}(\cN(\sigma))$. 
\end{corollary}
\begin{proof}
    This follows by choosing $f(x)=x\log(x)$ in~\cref{prp:f_div_Contraction} and rearranging the result.
\end{proof}
There are only few bounds in the literature that apply to $\delta>0$. A recent point of comparison can be found in~\cite[Theorem 6]{dasgupta2025quantum}, which shows for classical probability distributions $P$ and $Q$ that, 
\begin{align}
    D(\cN(P)\|\cN(Q)) \leq &\left( \epsilon\tanh(\frac\epsilon2) + \delta\left(\frac{2\epsilon}{e^\epsilon+1}+\frac{e^\epsilon+\delta-1}{e^\epsilon} + \log\frac1{1-\delta}\right)\right)\tau \nonumber\\
    &+ \delta\left( \frac{e^\epsilon}{1-\delta} + 2\log\frac{e^\epsilon}{1-\delta} - \frac{1-\delta}{e^\epsilon} + 2\left( \epsilon + \log\frac1{1-\delta} +\frac2m \right), \right)\label{Eq:dasgupta}
\end{align}
where $m\equiv m(\cN,Q,P)$ involves a truncated distribution defined in~\cite{dasgupta2025quantum}. We always have $m\leq\lambda$ as a point of comparison. For better comparison, we can write out the $\tanh$ function and reformulate~\cref{Eq:dasgupta} as
\begin{align}
    D(\cN(P)\|\cN(Q)) \leq &\left( \epsilon\frac{e^\epsilon-1+2\delta}{e^\epsilon+1} + \delta\left(\frac{e^\epsilon+\delta-1}{e^\epsilon} + \log\frac1{1-\delta}\right)\right)\tau \nonumber\\
    &+ \delta\left( \frac{e^\epsilon}{1-\delta}  - \frac{1-\delta}{e^\epsilon} + 4\left( \epsilon + \log\frac1{1-\delta} +\frac1m \right) \right).
\end{align}
Comparing to~\cref{Cor:RE-LDP-bound}, we see that the prefactor of $\tau$ is strictly better in our result. For the overall bounds, we provide a numerical comparison in Figure~\ref{fig:rev-pin-LDP}. For all examples we tried, our new bound performs significantly better (even in the setting $\lambda=m$, which provides a lower bound on the bound given in~\eqref{Eq:dasgupta} since $\lambda \leq m$).

\begin{figure} 
    \centering
\includegraphics[width=0.48\linewidth]{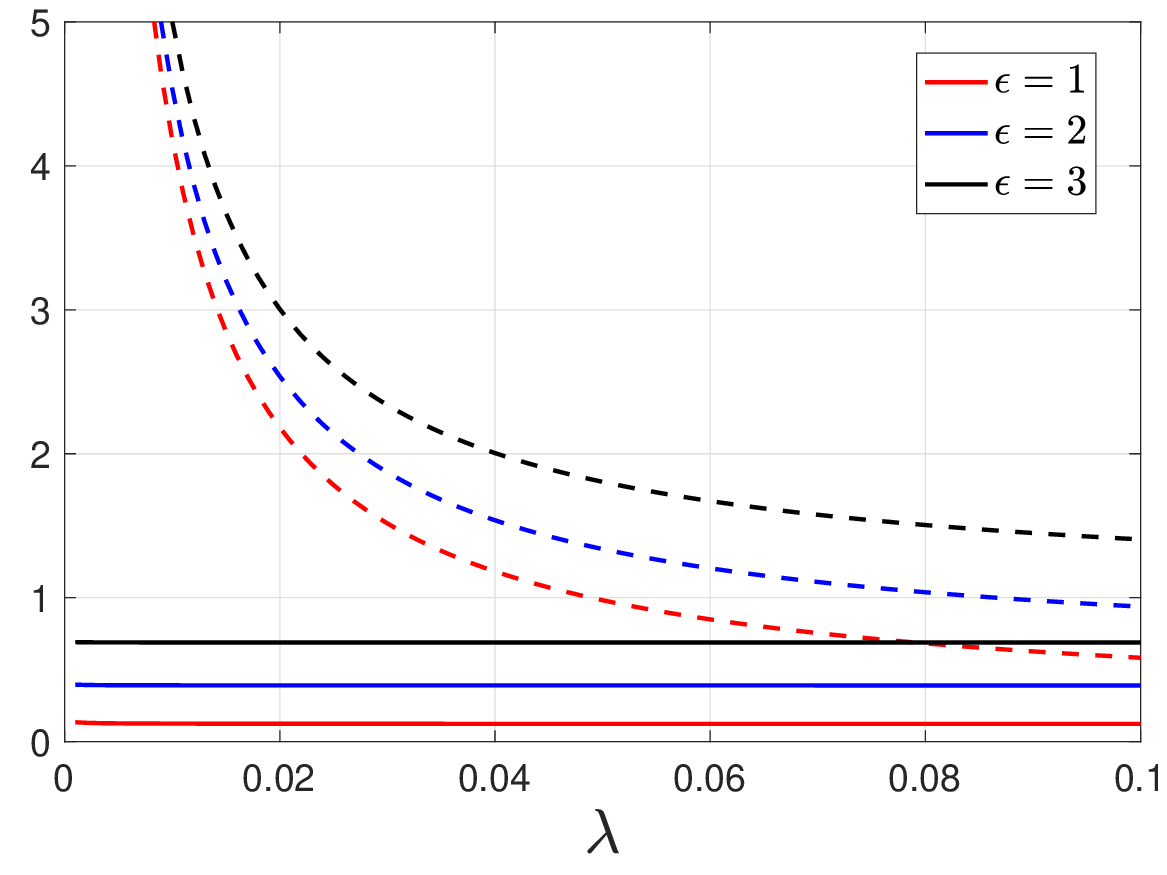}
\includegraphics[width=0.48\linewidth]{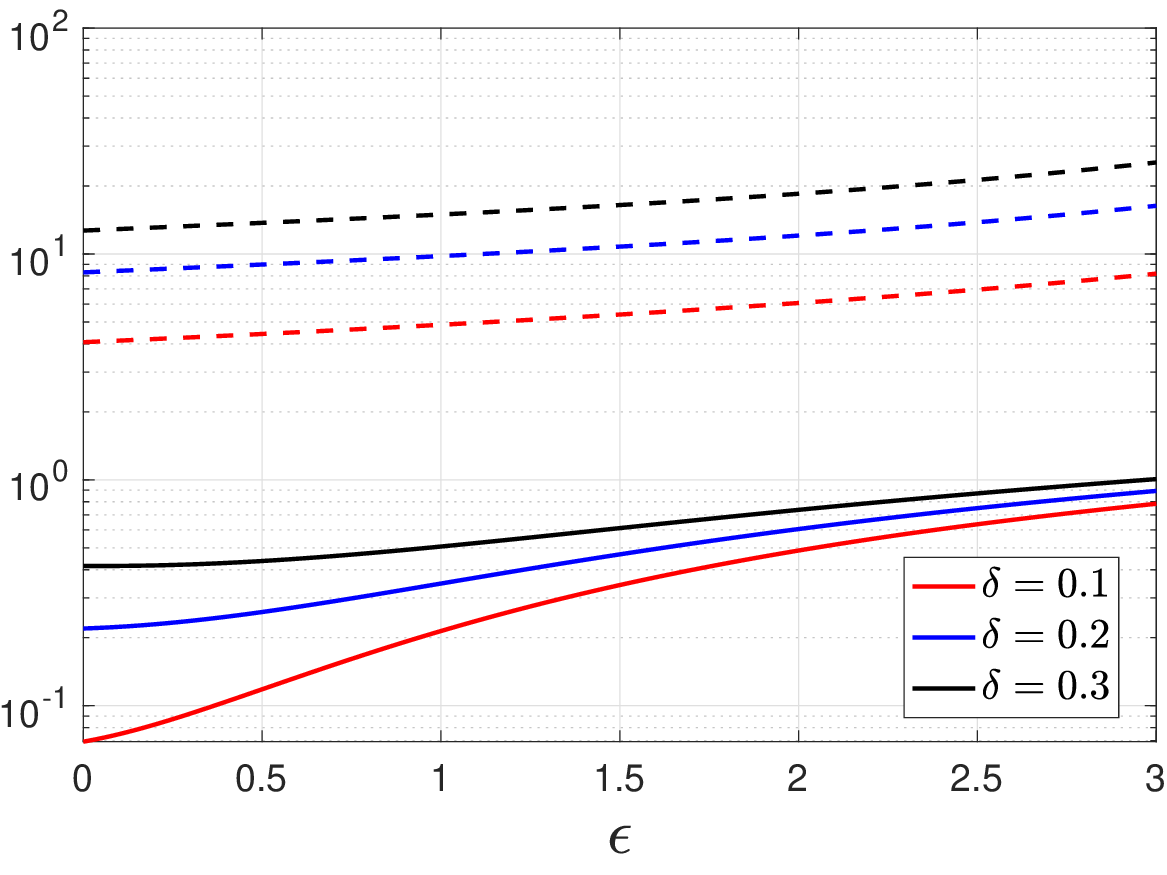}
    \caption{Comparing LDP bounds on the relative entropy. Dashed lines represent Equation~\eqref{Eq:dasgupta} (in particular a lower bound on that since we chose $\lambda=m$) and solid lines our new bound in Equation~\eqref{Eq:RE-LDP-bound}. Left: Plot over $\lambda$, respectively $m$, for fixed $\epsilon=\{1,2,3\}, \delta=0.01, \tau=0.25$. Right: Plot over $\epsilon$ for fixed $\delta=\{0.1,0.2,0.3\},  \lambda=m=0.1, \tau=0.25$.}
    \label{fig:rev-pin-LDP}
\end{figure}

\section{Conclusion and Future Work}
In this work, we study strong data-processing inequalities (SDPI) for quantum hockey-stick divergence. We obtain both linear and non-linear SDPI for quantum hockey-stick divergences, generalizing and improving upon the best known linear bounds in both the classical and quantum settings when the channel satisfies a certain criterion given in~\eqref{eq:intro_class_of_channels}. In fact, the non-linear SDPI in~\cref{thm:non_linear_HS_div} is tight as well as the first of its kind in the quantum setting.
Then, we analyze the setting where the composition of noisy channels is applied by defining $F_\gamma$ curves that reduce to Dobrushin curves for $\gamma=1$. We also studied in which cases the required criterion on the hockey-stick divergence holds for these bounds to be valid. Furthermore, we also obtained reverse Pinsker-type inequalities for $f$-divergences with additional constraints on hockey-stick divergences that would be of independent interest. Finally, we showed how the established results are useful in establishing finite, tighter bounds on mixing times of channels and stronger composition results with sequential composition of noisy quantum channels in ensuring privacy for quantum systems where privacy is quantified by quantum local differential privacy.

This work establishes a foundation to study non-linear SDPI for other quantum divergences, including quantum relative entropy and other families of $f$-divergences, possibly with the tools established in~\cite{hirche2024quantumDivergences}. In~\cref{rem:measured_HS_SDPI}, we showed how some of the results here generalize to the setting of measured hockey-stick divergences and how measured hockey-stick divergences can encode practical restrictions imposed by the systems. With that, it is an important research direction to explore SDPI for measured hockey-stick divergences and measured quantities in general. Another interesting exploration is to study non-linear SDPI under several operational constraints, similar to the study of energy-constrained Dobrushin curves studied in~\cite[Section 4]{huber2019jointly}.

\section*{Acknowledgments}
We thank Mark M.~Wilde for helpful discussions on strong data-processing inequalities. 
We also thank Behnoosh Zamanlooy, Shahab Asoodeh, Mario Diaz, and Flavio Calmon for sharing the extended version of their article~\cite{zamanlooy2024mathrm}, where we got inspired to derive quantum generalizations of their results. 
TN acknowledges support from the
Department of Mathematics and the IQUIST Postdoctoral Fellowship from
the Illinois Quantum Information Science and Technology Center at
the University of Illinois Urbana-Champaign. TN also acknowledges the support from the Dieter Schwarz Exchange Programme on Quantum Communication and Security at the Centre for Quantum Technologies, National University of Singapore, Singapore, during her research visit, and the hospitality of Marco Tomamichel's group, where the initial idea generation of this project happened.
IG is supported by the Ministry of Education, Singapore, through grant T2EP20124-0005. CH received funding by the Deutsche Forschungsgemeinschaft (DFG, German
Research Foundation) – 550206990. This work was supported, in part, by the Federal Ministry of Research, Technology and Space (BMFTR), Germany, under the QC service center QUICS (grant no. 13N17418).

\bibliographystyle{quantum}
\bibliography{lib}

\begin{thebibliography}{10}

\bibitem{Dobrushin1956}
R.~L. Dobrushin.
\newblock ``Central limit theorem for nonstationary {M}arkov chains. {I}''.
\newblock \href{https://dx.doi.org/10.1137/1101006}{Theory of Probability \& Its Applications {\bf 1}, 65--80}~(1956).

\bibitem{Hiai15}
Fumio Hiai and Mary~Beth Ruskai.
\newblock ``Contraction coefficients for noisy quantum channels''.
\newblock \href{https://dx.doi.org/10.1063/1.4936215}{Journal of Mathematical Physics {\bf 57}, 015211}~(2015).

\bibitem{Hirche2022contraction}
Christoph Hirche, Cambyse Rouz{\'{e}}, and Daniel Stilck~Fran{\c{c}}a.
\newblock ``On contraction coefficients, partial orders and approximation of capacities for quantum channels''.
\newblock \href{https://dx.doi.org/10.22331/q-2022-11-28-862}{{Quantum} {\bf 6}, 862}~(2022).

\bibitem{gao2022complete}
Li~Gao and Cambyse Rouz{\'e}.
\newblock ``Complete entropic inequalities for quantum {M}arkov chains''.
\newblock \href{https://dx.doi.org/10.1007/s00205-022-01785-1}{Archive for Rational Mechanics and Analysis {\bf 245}, 183--238}~(2022).

\bibitem{asoodeh2023contractionegammadivergenceapplicationsprivacy}
Shahab Asoodeh, Mario Diaz, and Flavio~P. Calmon.
\newblock ``Contraction of $e_\gamma$-divergence and its applications to privacy''~(2023).
\newblock  \href{http://arxiv.org/abs/2012.11035}{arXiv:2012.11035}.

\bibitem{nuradha2024contraction}
Theshani Nuradha and Mark~M. Wilde.
\newblock ``Contraction of private quantum channels and private quantum hypothesis testing''.
\newblock \href{https://dx.doi.org/10.1109/TIT.2025.3527859}{IEEE Transactions on Information Theory {\bf 71}, 1851--1873}~(2025).

\bibitem{hirche2025partial}
Christoph Hirche and Oxana Shaya.
\newblock ``{Partial orders and contraction for BISO channels}''.
\newblock In 2025 IEEE International Symposium on Information Theory (ISIT).
\newblock \href{https://dx.doi.org/10.1109/ISIT63088.2025.11195255}{Pages 1--6}.
\newblock IEEE~(2025).

\bibitem{george2025unifiedapproachquantumcontraction}
Ian George and Marco Tomamichel.
\newblock ``A unified approach to quantum contraction and correlation coefficients''~(2025).
\newblock  \href{http://arxiv.org/abs/2505.15281}{arXiv:2505.15281}.

\bibitem{George-2024ergodic}
Ian George, Alice Zheng, and Akshay Bansal.
\newblock ``Divergence inequalities with applications in ergodic theory''~(2024).
\newblock  \href{http://arxiv.org/abs/2411.17241}{arXiv:2411.17241}.

\bibitem{delsol2025computationalaspectstracenorm}
Idris Delsol, Omar Fawzi, Jan Kochanowski, and Akshay Ramachandran.
\newblock ``Computational aspects of the trace norm contraction coefficient''~(2025).
\newblock  \href{http://arxiv.org/abs/2507.16737}{arXiv:2507.16737}.

\bibitem{Doeblin1937}
W.~Doeblin.
\newblock ``Sur les proprietes asymptotiques de mouvement r\'egis par certains types de chaines simples''.
\newblock Bulletin mathematique de la Societe Roumaine des Sciences {\bf 39}, 57--115~(1937).
\newblock  url:~\url{http://www.jstor.org/stable/43769809}.

\bibitem{makur2024doeblin}
Anuran Makur and Japneet Singh.
\newblock ``Doeblin coefficients and related measures''.
\newblock \href{https://dx.doi.org/10.1109/TIT.2024.3367856}{IEEE Transactions on Information Theory {\bf 70}, 4667--4692}~(2024).

\bibitem{hirche2024quantum}
Christoph Hirche.
\newblock ``Quantum {D}oeblin coefficients: A simple upper bound on contraction coefficients''~(2024).
\newblock  \href{http://arxiv.org/abs/2405.00105v2}{arXiv:2405.00105v2}.

\bibitem{george2025quantumdoeblincoefficientsinterpretations}
Ian George, Christoph Hirche, Theshani Nuradha, and Mark~M. Wilde.
\newblock ``Quantum doeblin coefficients: Interpretations and applications''~(2025).
\newblock  \href{http://arxiv.org/abs/2503.22823}{arXiv:2503.22823}.

\bibitem{Polyanskiy-2015a}
Yury Polyanskiy and Yihong Wu.
\newblock ``Dissipation of information in channels with input constraints''.
\newblock \href{https://dx.doi.org/10.1109/TIT.2015.2482978}{IEEE Transactions on Information Theory {\bf 62}, 35--55}~(2015).

\bibitem{du2017strong}
Flavio du~Pin~Calmon, Yury Polyanskiy, and Yihong Wu.
\newblock ``Strong data processing inequalities for input constrained additive noise channels''.
\newblock \href{https://dx.doi.org/10.1109/TIT.2017.2782359}{IEEE Transactions on Information Theory {\bf 64}, 1879--1892}~(2017).

\bibitem{lu2024doeblin}
William Lu, Anuran Makur, and Japneet Singh.
\newblock ``On doeblin curves and their properties''.
\newblock In 2024 IEEE International Symposium on Information Theory (ISIT).
\newblock \href{https://dx.doi.org/10.1109/ISIT57864.2024.10619264}{Pages 2544--2549}.
\newblock IEEE~(2024).

\bibitem{zamanlooy2024mathrm}
Behnoosh Zamanlooy, Shahab Asoodeh, Mario Diaz, and Flavio~P Calmon.
\newblock ``${E}_\gamma$-mixing time''.
\newblock In 2024 IEEE International Symposium on Information Theory (ISIT).
\newblock \href{https://dx.doi.org/10.1109/ISIT57864.2024.10619250}{Pages 3474--3479}.
\newblock IEEE~(2024).

\bibitem{harremoes2011pairs}
Peter Harremo{\"e}s and Igor Vajda.
\newblock ``On pairs of $ f $-divergences and their joint range''.
\newblock \href{https://dx.doi.org/10.1109/TIT.2011.2137353}{IEEE Transactions on Information Theory {\bf 57}, 3230--3235}~(2011).

\bibitem{huber2019jointly}
Stefan Huber, Robert K{\"o}nig, and Marco Tomamichel.
\newblock ``Jointly constrained semidefinite bilinear programming with an application to {D}obrushin curves''.
\newblock \href{https://dx.doi.org/10.1109/TIT.2019.2939474}{IEEE Transactions on Information Theory {\bf 66}, 2934--2950}~(2019).

\bibitem{sharma2012strongconversesquantumchannel}
Naresh Sharma and Naqueeb~Ahmad Warsi.
\newblock ``On the strong converses for the quantum channel capacity theorems''~(2012).
\newblock  \href{http://arxiv.org/abs/1205.1712}{arXiv:1205.1712}.

\bibitem{hirche2024quantumDivergences}
Christoph Hirche and Marco Tomamichel.
\newblock ``Quantum {R}\'enyi and $f$-divergences from integral representations''.
\newblock \href{https://dx.doi.org/10.1007/s00220-024-05087-3}{Communications in Mathematical Physics {\bf 405}, 208}~(2024).

\bibitem{hirche2022quantum}
Christoph Hirche, Cambyse Rouz{\'e}, and Daniel~Stilck Fran{\c{c}}a.
\newblock ``Quantum differential privacy: An information theory perspective''.
\newblock \href{https://dx.doi.org/10.1109/TIT.2023.3272904}{IEEE Transactions on Information Theory {\bf 69}, 5771--5787}~(2023).
\newblock  \href{http://arxiv.org/abs/2202.10717}{arXiv:2202.10717}.

\bibitem{nuradha2023quantum}
Theshani Nuradha, Ziv Goldfeld, and Mark~M. Wilde.
\newblock ``Quantum pufferfish privacy: A flexible privacy framework for quantum systems''.
\newblock \href{https://dx.doi.org/10.1109/TIT.2024.3404927}{IEEE Transactions on Information Theory {\bf 70}, 5731--5762}~(2024).

\bibitem{angrisani2023differentialprivacyamplificationquantum}
Armando Angrisani, Mina Doosti, and Elham Kashefi.
\newblock ``Differential privacy amplification in quantum and quantum-inspired algorithms''~(2023).
\newblock  \href{http://arxiv.org/abs/2203.03604}{arXiv:2203.03604}.

\bibitem{dasgupta2025quantum}
Ayanava Dasgupta, Naqueeb~Ahmad Warsi, and Masahito Hayashi.
\newblock ``Quantum information ordering and differential privacy''~(2025).
\newblock  \href{http://arxiv.org/abs/2511.01467}{arXiv:2511.01467}.

\bibitem{gallage2025theory}
Theshani Nuradha~Piliththuwasam Gallage.
\newblock ``Theory of privacy and testing in a quantum world''.
\newblock PhD thesis.
\newblock Cornell University.
\newblock ~(2025).
\newblock  url:~\url{https://www.proquest.com/openview/2cb580e718241481b84185ba6d289ce7/}.

\bibitem{zamanlooy2023strong}
Behnoosh Zamanlooy and Shahab Asoodeh.
\newblock ``Strong data processing inequalities for locally differentially private mechanisms''.
\newblock In 2023 IEEE International Symposium on Information Theory (ISIT).
\newblock \href{https://dx.doi.org/10.1109/ISIT54713.2023.10206578}{Pages 1794--1799}.
\newblock IEEE~(2023).

\bibitem{regula2025tight}
Bartosz Regula, Ludovico Lami, and Nilanjana Datta.
\newblock ``Tight relations and equivalences between smooth relative entropies''~(2025).
\newblock  \href{http://arxiv.org/abs/2501.12447}{arXiv:2501.12447}.

\bibitem{Christoph2024sample}
Hao-Chung Cheng, Christoph Hirche, and Cambyse Rouz{\'e}.
\newblock ``Sample complexity of locally differentially private quantum hypothesis testing''~(2024).
\newblock  \href{http://arxiv.org/abs/2406.18658}{arXiv:2406.18658}.

\bibitem{sason2015reverse}
Igal Sason.
\newblock ``On reverse pinsker inequalities''~(2015).
\newblock  \href{http://arxiv.org/abs/1503.07118}{arXiv:1503.07118}.

\bibitem{nuradha2025MeasuredHS}
Theshani Nuradha, Vishal Singh, and Mark~M. Wilde.
\newblock ``Measured hockey-stick divergence and its applications to quantum pufferfish privacy''.
\newblock In 2025 IEEE International Symposium on Information Theory (ISIT).
\newblock \href{https://dx.doi.org/10.1109/ISIT63088.2025.11195501}{Pages 1--6}.
\newblock ~(2025).

\bibitem{beigi2025some}
Salman Beigi, Christoph Hirche, and Marco Tomamichel.
\newblock ``Some properties and applications of the new quantum $ f $-divergences''~(2025).
\newblock  \href{http://arxiv.org/abs/2501.03799}{arXiv:2501.03799}.

\bibitem{liu2025layer}
Po-Chieh Liu, Christoph Hirche, and Hao-Chung Cheng.
\newblock ``Layer cake representations for quantum divergences''~(2025).
\newblock  \href{http://arxiv.org/abs/2507.07065}{arXiv:2507.07065}.

\bibitem{lanier2025classical}
Dimitri Lanier, Julien B{\'e}guinot, and Olivier Rioul.
\newblock ``From classical to quantum: Explicit classical distributions achieving maximal quantum $ f $-divergence''~(2025) \href{http://arxiv.org/abs/2501.14340}{arXiv:2501.14340}.

\bibitem{binette2019note}
Olivier Binette.
\newblock ``A note on reverse pinsker inequalities''.
\newblock \href{https://dx.doi.org/10.1109/TIT.2019.2896192}{IEEE transactions on information theory {\bf 65}, 4094--4096}~(2019).

\bibitem{levin2017markov}
David~A. Levin and Yuval Peres.
\newblock ``{M}arkov chains and mixing times''.
\newblock \href{https://dx.doi.org/978-1-4704-2962-1}{Volume 107}.
\newblock American Mathematical Society. ~(2017).

\bibitem{Raginsky2016}
Maxim Raginsky.
\newblock ``Strong data processing inequalities and {$\Phi $}-{S}obolev inequalities for discrete channels''.
\newblock \href{https://dx.doi.org/10.1109/TIT.2016.2549542}{IEEE Transactions on Information Theory {\bf 62}, 3355--3389}~(2016).

\bibitem{wolf2012quantum}
Michael~M. Wolf.
\newblock ``Quantum channels and operations---guided tour''~(2012).
\newblock Available at \url{https://mediatum.ub.tum.de/doc/1701036/document.pdf}.

\bibitem{QDP_computation17}
Li~Zhou and Mingsheng Ying.
\newblock ``Differential privacy in quantum computation''.
\newblock In Proceedings of IEEE Computer Security Foundations Symposium (CSF).
\newblock \href{https://dx.doi.org/10.1109/CSF.2017.23}{Pages 249--262}.
\newblock IEEE~(2017).

\bibitem{aaronson2019gentle}
Scott Aaronson and Guy~N. Rothblum.
\newblock ``Gentle measurement of quantum states and differential privacy''.
\newblock In Proceedings of ACM SIGACT Symposium on Theory of Computing.
\newblock \href{https://dx.doi.org/10.1145/3313276.331637}{Pages 322--333}.
\newblock ~(2019).

\bibitem{DMNS06}
Cynthia Dwork, Frank McSherry, Kobbi Nissim, and Adam Smith.
\newblock ``Calibrating noise to sensitivity in private data analysis''.
\newblock In Proceedings of Conference on Theory of Cryptography, TCC.
\newblock \href{https://dx.doi.org/10.1109/ITW.2014.6970882}{Pages 265--284}.
\newblock ~(2006).

\bibitem{DR14}
Cynthia Dwork and Aaron Roth.
\newblock ``The algorithmic foundations of differential privacy''.
\newblock Foundations and Trends in Theoretical Computer Science (FnT-TCS) {\bf 9}, 211--407~(2014).
\newblock  url:~\url{10.1561/0400000042}.

\bibitem{KM14}
D.~{Kifer} and A.~{Machanavajjhala}.
\newblock ``Pufferfish: A framework for mathematical privacy definitions''.
\newblock \href{https://dx.doi.org/10.1145/2514689}{ACM Transactions on Database Systems {\bf 39}, 1--36}~(2014).

\bibitem{nuradha2022pufferfishJ}
Theshani Nuradha and Ziv Goldfeld.
\newblock ``Pufferfish privacy: An information-theoretic study''.
\newblock \href{https://dx.doi.org/10.1109/TIT.2023.3296288}{IEEE Transactions on Information Theory {\bf 69}, 7336--7356}~(2023).

\bibitem{angrisani_localModel25}
Armando Angrisani and Elham Kashefi.
\newblock ``Quantum differential privacy in the local model''.
\newblock \href{https://dx.doi.org/10.1109/TIT.2025.3552671}{IEEE Transactions on Information Theory {\bf 71}, 3675--3692}~(2025).

\bibitem{guan2024optimal}
Ji~Guan.
\newblock ``Optimal mechanisms for quantum local differential privacy''~(2024).
\newblock  \href{http://arxiv.org/abs/2407.13516}{arXiv:2407.13516}.

\end{thebibliography}
\end{document}